\documentclass[preprint,12pt]{elsarticle}

\usepackage{amssymb}

\journal{Information and Computation}

\usepackage{microtype}

\usepackage[utf8]{inputenc}
\usepackage{amsfonts}
\usepackage{amsmath, amssymb}
\usepackage{amsthm}
\usepackage{graphicx}
\usepackage{wrapfig}
\usepackage{tikz}
\usetikzlibrary{shapes,decorations,calc}
\usepackage{todonotes}
\usepackage{xspace}
\usepackage{subcaption}
\usepackage{mathdots}

\usepackage{mathabx}

\usepackage[shortcuts]{extdash} 

\usepackage{color}
\usepackage{algorithm}
\usepackage{algorithmic}

\newcommand{\treeZ}{t}

\newcommand{\eqdef}{\stackrel{\text{def}}=}
\newcommand{\eqext}[1]{\stackrel{#1}=}

\newcommand{\comment}[1]{}

\newcommand{\dL}{\mathtt{\scriptstyle L}}
\newcommand{\dR}{\mathtt{\scriptstyle R}}

\newcommand{\dS}{\{\dL, \dR\}\xspace}

\newcommand{\setD}{\dS}
\newcommand{\setPositions}{\{ \dL, \dR \}^{\ast}}
\newcommand{\setNodes}{\dS^{\ast}}

\newcommand{\powerset}{\mathcal{P}}

\newcommand{\init}{\mathrm{I}}
\newcommand{\lang}{\mathrm{L}}

\newcommand{\graph}{\mathrm{G}}

\newcommand{\blank}{\mbox{$\flat$}}

\newcommand{\bbN}{\mathbb{N}}
\newcommand{\bbR}{\mathbb{R}}

\renewcommand{\AA}{\mbox{$\mathcal{A}$}}

\newcommand{\nhasBlank}[1] {#1}

\renewcommand{\AA}{\mathcal{A}}

\newcommand{\CC}{\mathcal{C}}
\newcommand{\DD}{\mathcal{D}}
\newcommand{\FF}{\mathcal{F}}
\newcommand{\MM}{\mathcal{M}}

\newcommand{\Uu}{\mathcal{U}}

\newcommand{\bB}{\mathbb{B}}
\newcommand{\bR}{\mathbb{R}}

\newcommand{\tL}{\mathrm{L}}
\newcommand{\td}{\mathrm{d}}

\usepackage{bbm}

\newcommand{\one}{\mathbbm{1}\xspace}

\newcommand{\w}{\omega}

\newcommand{\trees}[1] {\mathcal{T}_{#1}\xspace}
\newcommand{\treesFinite}[1] {\mathcal{T}_{#1}^{{<}\omega}\xspace}
\newcommand{\treesInfinite}[1] {\mathcal{T}_{#1}^{\omega}\xspace}
\newcommand{\setBall}[1]{\mathbb{B}_{#1}}
\newcommand{\setBallInNode}[2]{\mathbb{B}_{#1,#2}}

\newcommand{\dom}[1] {\textit{Dom}(#1)}

\newcommand{\nodes}[1] {\dom{#1}}

\newcommand{\fun}[3]{\ensuremath{#1\colon #2 \to #3}}
\newcommand{\parfun}[3]{\ensuremath{#1\colon #2 \rightharpoonup #3}}

\newcommand{\rootP}{\varepsilon}

\newcommand{\ancestor}{\sqsubsetneq}

\newcommand{\weakAncestor}{\sqsubseteq}

\newcommand{\subtreeSymbol}{.}

\newcommand{\isPrefix}[2]{#1{\sqsubseteq}#2}
\newcommand{\prefix}{\sqsubseteq}

\newcommand{\subtreeAtNode}[2]{{#1}{\subtreeSymbol}#2}

\newcommand{\decpro}[4]{
\begin{problem}[#1]
	\label{#2}
	\[\begin{tabular}{r p{10cm}}
		Input:  & #3 \\
		Output: & #4 \\
	\end{tabular}
	\]
\end{problem}
}

\newcommand{\nexp}{\mbox{\textit{NEXP}}}
\newcommand{\expspace}{\mbox{\textit{EXPSPACE}}}

\newcommand{\exptime}{\textit{EXP}}

\newcommand{\np}{\textit{NP}}
\newcommand{\aptime}{\textit{APTIME}\xspace}

\newcommand{\standardMeasureS}{\mu^{\ast}}
\newcommand{\standardMeasure}[1]{\standardMeasureS(#1)}
\newcommand{\standardMeasureBig}[1]{\standardMeasureS\big(#1\big)}

\newcommand{\type}{\mathrm{type}\xspace}

\tikzstyle{eve}       = [scale=0.8, inner sep=0, shape=diamond,draw=black,minimum size=10mm]
\tikzstyle{adam}      = [scale=0.8, inner sep=0, shape=rectangle,draw=black,minimum size=8mm]
\tikzstyle{branching} = [scale=0.8, inner sep=0, regular polygon, regular polygon sides=3 ,draw=black,minimum size=10mm]
\tikzstyle{generic} = [scale=0.8, inner sep=0, regular polygon, regular polygon sides=8 ,draw=black,minimum size=10mm]
\tikzstyle{nature}    = [scale=0.8, circle, draw=black, minimum size=9mm]

\tikzstyle{successor} = [->, draw=black!80]

\newtheorem{thm}{Theorem}[section]
\newtheorem{lemma}[thm]{Lemma}

\newtheorem{problem}[thm]{Problem}
\newtheorem{proposition}[thm]{Proposition}
\newtheorem{corollary}[thm]{Corollary}
\newtheorem{claim}[thm]{Claim}
\newtheorem{fact}[thm]{Fact}

\newtheorem{example}[thm]{Example}
\newtheorem{remark}{Remark}[section]

\usepackage{mathabx}

\usepackage{listings}
\date{October 29, 2019}

\begin{document}




\begin{frontmatter}




\title{The Uniform Measure of Simple Regular Sets\\ of Infinite Trees\footnote{Both authors were supported by Poland's National Science Centre grant no.~2016/21/D/ST6/00491.}}

\address[label1]{University of Warsaw}
\address[label2]{University of New Caledonia}

\author[label1]{Michał Skrzypczak}
\author[label1,label2]{Marcin Przybyłko}

\address{}

\begin{abstract}
    We consider the problem of~computing the measure of~a~regular set of~infinite binary trees.
    While the general case remains unsolved, we~show that the measure of~a~language can be~computed when the set is~given in~one of~the following three formalisms: a~first-order
    formula with no~descendant relation; a~Boolean combination of~conjunctive queries (with descendant relation); or~by~a~non\=/deterministic safety tree automaton. Additionally, in~the first two cases the measure of~the set is~always rational, while in~the third it~is~an~algebraic number. Moreover, we~provide an~example of a~first-order formula that uses descendant
    relation and defines a~language of~infinite trees having an~irrational (but algebraic) measure. 
\end{abstract}


\end{frontmatter}



\section{Introduction}
\label{sec:intro}

The problem of~computing a~measure of~a~set can be seen as~one of~the fundamental problems considered in~the study of~probabilistic systems.
This problem has been studied mostly implicitly, as~it~is~often one of~the intermediary steps in~solving stochastic games, cf.~\cite{chatterjeeStochasticRegular},
in~answering queries in~probabilistic databases, cf.~\cite{suciuProbDatabases}, or~in~model checking for stochastic branching processes, cf.~\cite{chen_model_checking}.

To~us, this problem naturally arises in~the study of~stochastic games.
Many of~the games considered in~the literature can be~seen as~instances of~the stochastic version of Gale\=/Stewart games~\cite{galeGames}.
Such games use winning conditions expressed as~sets of~winning plays, i.e.~a~set of~(in)finite words representing winning plays.
Hence, computing the value of~a~stochastic game involves computing the measure of~the winning set with respect to~the probabilistic space generated
by~the stochastic elements of~the game. 

Mio~\cite{mioBranchingGames} introduced branching games, i.e.~stochastic games for which
the plays are represented as~(in)finite trees, rather than words.
Then, the questions concerning whether a~set of~trees has a~well\=/defined measure~\cite{michalewskiMeasureFinal},
and whether that measure can be~computed~\cite{michalewskiCompMeasure} have been raised and partially answered.

\comment{
\michal{Czy mozemy wywalic ponizszy paragraf? I tak jest potem powtorzony...}
\marcin{Tak, jesli chcesz mozes sporo tu pozmieniac}
In~this work we~focus our attention on~the problem of~computing the uniform measure of~a~given set of~infinite binary trees in three specific  cases: when the set is~defined by~a~first\=/order formula, when the set is defined by~a~Boolean combination of~conjunctive queries, or~ when it is defined by~a~non\=/deterministic safety tree automaton.
}

\paragraph{Related work}
The problem of computing the measure of an arbitrary regular set of~trees has been already, explicitly or~implicitly, studied.
Gogacz et al.~\cite{michalewskiMeasureFinal} prove that regular sets~of trees are universally measurable.
In the case of infinite trees,
Chen et~al.~\cite{chen_model_checking} show that the measure of a~set accepted by~deterministic automaton is~computable;
Michalewski and Mio~\cite{michalewskiCompMeasure} extend the class of~sets with computable measure to~the class of~sets defined by the so\=/called game automata.
In~the case of~finite trees, Amarilli et~al.~\cite{amarilliProvenance} show that with measures defined by~fragments of~the probabilistic XML, i.e.~where the support of~the measure consists of~the trees of~bounded depth, the measure is~computable for arbitrary regular sets of~trees. 
In~the case of~regular languages of~infinite words, Staiger~\cite{staigerMesureOnWords} shows that the measure of~every regular language of~words is~computable.
Note that, in~all the above results, the inherent deterministic nature of~the involved automata plays an~important role.

The problem of~computing the measure of~a~set of~infinite trees has also been, implicitly, considered in~probability games.
The problem is~a~special case of~computing the value of~a~stochastic game when the strategies of~players are already chosen.
In~the case of~infinite trees, Przybyłko and Skrzypczak~\cite{przybylkoRegularBranchingGamesMFCS} consider branching games with regular wining sets.
In the case of words, for the survey of probabilistic $\omega$\=/regular games on graphs see e.g.~Chatterjee and Henzinger~\cite{chatterjeeStochasticRegular}.

The problem under consideration can be~also seen as~the problem of~query evaluation in~the probabilistic databases setting. For instance Amarilli et~al.~\cite{amarilliCQsOnProbabilisticGraphs} enquire into the evaluation problem of~conjunctive queries over probabilistic graphs.
For an~introduction to~probabilistic databases see e.g.~\cite{suciuProbDatabases}.

\paragraph{Our contribution}
We~provide algorithms that~compute the measure of~tree sets belonging to~some restricted classes of~regular languages.
We~show that, in~the case of first\=/order (FO) formulae using unary predicates and child relation, the uniform measure can be computed in~three\=/fold exponential space.
We also show that in~the case of~Boolean combinations of~conjunctive queries (BCCQ) using unary predicates, child relation, and descendant relation, the measure can be~computed in~exponential space. Additionally, we provide an algorithm for computing the measure of a language given by a non\=/deterministic safety tree automaton. As the class of languages recognisable by these automata coincides with the class of topologically closed regular languages (see Proposition~\ref{pro:closed-is-safety}), this result provides a~method for languages at the basic level of topological complexity of regular tree languages. The algorithm translates the structure of a~given automaton $\AA$ into an~exponentially bigger first\=/order formula over the field of reals ($\bbR$). Thus, decision problems about the measure of $\lang(\AA)$ can be solved in doubly exponential space.

We~additionally provide an~example of~a~first\=/order formula over a~two letter alphabet
for which the defined set of trees has an~irrational uniform measure.
An~example of~a~regular language with an~irrational uniform measure was already presented in~\cite{michalewskiCompMeasure}, however that language is~not first\=/order definable.


This paper is an~extended journal version of~\cite{przybylko_tree_games}. Parts of the material presented here have been included in the PhD thesis of the first author~\cite{przybylko_phd}.

\paragraph{Organization of rest of the paper}
In~Section~\ref{sec:prelim} we~define basic notions used in~this article.
In~Section~\ref{sec:examples} we~showcase some basic properties of~the uniform measure on~selected examples.
The computability of~the measure of~the regular languages defined by~first\=/order
formulae is~discussed in~Section~\ref{sec:fo-sets}.
The computability of the measure of~the regular languages defined by~conjunctive queries is~discussed in~Section~\ref{sec:CQsMeasure}.
In Section~\ref{sec:open-sets} we~discuss the case of~safety automata.
Finally, in Section~\ref{sec:complexity} we provide some computational complexity bounds.
In the last section, we summarise the obtained results and propose some directions of future research.

\paragraph{Acknowledgements}

The authors would like to express their gratitude to Damian Niwi{\'n}ski for a~number of insightful comments on the topic. Also, the authors thank the anonymous referees for their careful reviews and helpful suggestions.

\section{Preliminaries}
\label{sec:prelim}

In this section we present crucial definitions used throughout this work.
We assume basic knowledge of~logic, automata, and measure.
For introduction to logic and automata see~\cite{thomasLanguages}, for introduction
to topology and measure see~\cite{kechrisDescriptive}.

\paragraph{Words and trees}
By~$\mathbb{N}$ we~denote the set of~natural numbers, i.e.~the set $\{0,1,2,\dots\}$ which, when treated as an~ordinal is also denoted by~$\omega$.
An~alphabet $\Gamma$ is~any non\=/empty finite set.
A~\emph{word} is~a~partial function \parfun{w}{\mathbb{N}}{\Gamma} such that the domain $\dom{w}$ of~$w$ is~$\leq$\=/closed. 
By $|w|$ we denote the length of~the word $w$, i.e. the size of~its domain.
By $\varepsilon$ we denote the empty word, i.e.~the unique word of~length $0$.
If the domain of~a~word $w$ is~finite, then the word is~called finite; otherwise, it is~called infinite.
Let $n \in \mathbb{N}$ and $\bowtie \in \{<, \leq, =\}$, then the set of~all words over an~alphabet $\Gamma$ of~length $l$ such that $l \bowtie n$  is~denoted~$\Gamma^{\bowtie n}$, e.g.~$\{0,1\}^{\leq 5}$ is~the set of~all binary words of~length at~most~$5$.
Following the usual convention, the set of~all finite words over an~alphabet $\Gamma$, i.e. the set $\Gamma^{<\omega}$, is~denoted $\Gamma^{\ast}$
and the set of~all infinite words over $\Gamma$, i.e.~the set $\Gamma^{=\omega}$, is~denoted $\Gamma^{\omega}$.

A~word $w$ is~called a~\emph{prefix} of~a~word $v$, denoted $w \sqsubseteq v$, if~$\dom{w} \subseteq \dom{v}$ and for every $i \in \dom{w}$ we have that
$w(i) = v(i)$. By $w \cdot v$, or~simply $wv$,  we~denote the \emph{concatenation} of~the words $w$ and $v$.

A~\emph{tree} is~any~partial function \parfun{t}{\setPositions}{\Gamma}, where the domain $\dom{t}$ is~a~non\=/empty prefix\=/closed set.
The elements of~the set $\setD$ are called \emph{directions} (\emph{left} and \emph{right}, respectively) and the elements of~the set $\setPositions$ are called \emph{positions}.

For a given tree $t$, the elements of~the set $\dom{t}$ are called nodes of~$t$, or~nodes for short.
Any tree $t$ is~either \emph{finite}, if~its domain $\dom{t}$ is~finite, or~\emph{infinite}.
A~tree $t$ is~called a~\emph{full tree of~height $k$} if~$\dom{t} = \setD^{\le k}$.
A~tree $t$ is~called a~\emph{full tree} if~$\dom{t} = \setPositions$.

Let $\Gamma$ be an alphabet. The set of~all trees over the alphabet $\Gamma$ is~denoted by $\trees{\Gamma}$;
the set of~all finite trees by $\treesFinite{\Gamma}$;
the set of~all full trees of~height $k$ by $\trees{\Gamma}^{k}$;
the set of~all full trees by $\treesInfinite{\Gamma}$.
A~tree $t_1$ is~called a~\emph{prefix} of~a~tree $t_2$, denoted $t_1 \weakAncestor t_2$, if~$\dom{t_1} \subseteq \dom{t_2}$ and for every $u \in \dom{t_1}$ we have that
$t_1(u) = t_2(u)$. 
For a~tree~$t$ and a~node~$u \in \nodes{t}$, by~$\subtreeAtNode{t}{u}$ we~denote the unique tree such that for every position $v \in \setPositions$ the following holds.

\begin{equation}
\label{eq:subtreeDefinition}
\subtreeAtNode{t}{u}(v) \eqdef t(uv)
\end{equation}

The tree $\subtreeAtNode{t}{u}$ is~called the \emph{sub\=/tree of~$t$ in~the node~$u$}.
For a~tree $\treeZ$ and a~position~$u$, by $\setBallInNode{\treeZ}{u}$ we denote the set of~all full trees in which $t$ is~a prefix of~the sub\=/tree in~the node $u$, i.e.
\begin{equation}
\label{eq:ball}
\setBallInNode{\treeZ}{u} \eqdef \{\treeZ' \in \treesInfinite{\nhasBlank{\Gamma}} \mid \isPrefix{\treeZ}{\subtreeAtNode{\treeZ'}{u}}\},
\end{equation} with $\mathbb{B}_{\treeZ} \eqdef \setBallInNode{\treeZ}{\varepsilon}$.

\paragraph{Logic}
A tree $t$ over an~alphabet $\Gamma$ can be seen as a~relational structure
$\langle \dom{t}, {\varepsilon}, {s_{\dL}}, {s_{\dR}}, {s}, {\ancestor}, (a^{t})_{a \in \Gamma} \rangle$, where

\begin{itemize}
    \item $\dom{t}$ is~the domain of~$t$;
    \item $\varepsilon$ is~the root constant;
    \item $s_{\dL},s_{\dR}  \subseteq \dom{t} \times \dom{t}$ are the left child relation ($u\ {s_\dL}\ u\cdot\dL$) and the right child relation ($u\ {s_\dR}\ u\cdot\dR$), respectively;
    \item $s$ is~the child relation $s_{\dL} \cup s_{\dR}$;
    \item $\ancestor$ is~the ancestor relation, i.e. the transitive closure of~the relation $s$;
    \item $a^{t} \subseteq \dom{t}$ is~a~subset of~$\dom{t}$, for ${a \in \Gamma}$, and the family of~sets  $(a^{t})_{a \in \Gamma}$ is~a partition of~$\dom{t}$.
\end{itemize}
The partition $(a^{t})_{a \in \Gamma}$ induces the tree $t$ in the natural way: ${t}(u) = a$ if~and only if~$u \in a^{t}$.

The \emph{distance} between two positions is~the function \fun{d}{\setPositions \times \setPositions}{\bbN} defined as  $\td(u,v) = |u| + |v| - 2|x|$, where
$x$ is~the longest common prefix of~$u$ and~$v$.
Equivalently, the distance between two different
positions is~the length of~the shortest undirected path connecting the two nodes in~the graph $\langle \setPositions,  s_{\dL} {\cup} s_{\dR}\rangle$.  

\paragraph{Regular languages}
Formulae of~monadic second\=/order logic (MSO) can quantify over nodes in trees $\exists x$, $\forall x$ and over sets of~nodes $\exists X$, $\forall X$.
A~first\=/order (FO) formula is~an~MSO formula that does not quantify over the sets of~nodes.
A~sentence is~a~formula with no free variables.

We say that an~MSO formula $\varphi$ is~over a~signature $\Sigma$ if~$\varphi$ is~a well\=/formed formula
built from the symbols in $\Sigma$ together with the quantifiers and logical connectives.
Let $\Gamma$ be an alphabet, in this paper, we consider only formulae over the signatures $\Sigma$
such that $\Sigma \subseteq \{\rootP, s_{\dL}, s_{\dR}, s, \ancestor\} \cup \Gamma$.

Let $\varphi$ be a~first\=/order formula.
We write $t,v \models \varphi(x_1, \dots, x_k)$, if~the tree $t$, as a~relational structure, with the~valuation $v \in \dom{t}^k$ satisfies the formula $\varphi(x_1,\dots, x_k)$.
If $\varphi$ is~a sentence, we simply write $t \models \varphi$.
We say that a~formula $\varphi(x_1, \dots, x_k)$ is~satisfiable if~there is~a tree $t$ and a~tuple $v \in \dom{t}^k$ such that $t,v \models \varphi(x_1,\ldots,x_k)$.

Let $\Gamma$ be an alphabet, the set defined by an~MSO sentence $\varphi$, denoted $\lang(\varphi)$,
 is~the set of~all full trees over the alphabet $\Gamma$ that satisfy $\varphi$, i.e.~$\lang(\varphi) \eqdef \{ t \in \treesInfinite{\Gamma} \mid t \models \varphi\}$.
Such a~set of trees  is~called regular. This definition of~regular languages of~trees is~equivalent to the automata based definition, cf. e.g.~\cite{thomasLanguages}.

\paragraph{Topology and measure}
Recall that the set of~all full trees over an~alphabet $\Gamma$, denoted $\treesInfinite{\Gamma}$, is~the set of~all functions \fun{t}{\setPositions}{\Gamma}.
This set can naturally be enhanced with a topology in such a way that it becomes a homeomorphic copy of~the Cantor set, see~Gogacz et al.~\cite{michalewskiMeasureFinal} for 
more detailed definitions.

Note that the family of~the sets of~the form
\begin{equation}
\label{eq:baseTopology}
\{\fun{t}{\setD^{\ast}}{\Gamma} \mid \tau \prefix t\},
\end{equation}
where $\tau \in \treesFinite{{\Gamma}}$ is~a~full tree of~some finite height, constitutes a~\emph{base of~that topology}. A~set from the basis is~called a~\emph{base set}. Notice that the above set equals $\setBall{\tau}$, see~\eqref{eq:ball}.

A~set is~\emph{open} if~it~is~a~union, possibly empty, of~some base sets;
\emph{closed} if~it~is~the~complement of~an~open set;
\emph{clopen} if~it~is~both open and closed.
Note that our chosen basis consists of~clopen sets. Additionally, the family of clopen sets is closed under finite Boolean combinations.

The \emph{uniform} measure $\mu^{\ast}$ defined on the set of~full trees $\treesInfinite{\Gamma}$
 is the unique complete probability Borel measure such that for every finite tree $\tau \in \treesFinite{\Gamma}$ we have that
$\standardMeasure{\bB_{\tau}} = |\Gamma|^{-|\nodes{\tau}|}$. In other words, this measure is~such that for every node $u \in \dS^{\ast}$ and label $a \in \Gamma$, 
the probability that in a random tree~$\treeZ$ the node $u$ is~labelled with the letter $a$ is~$\frac{1}{|\Gamma|}$, i.e.~$\standardMeasure{\{ t \in \treesInfinite{\Gamma} \mid t(u)=a \}} = \frac{1}{|\Gamma|}$.
Notice that for any two distinct positions $u$, $v$ and two letters $a$, $b$ the events $S_{u,a} = \{ t \in \treesInfinite{\Gamma} \mid t(u)=a \}$ and $S_{v,b} = \{ t \in \treesInfinite{\Gamma} \mid t(v)=b \}$ are independent, i.e. 
\begin{equation}
\label{eq:indepencence}
\standardMeasureS(S_{u,a} \cap S_{v,b}) = \standardMeasureS(S_{u,a}) \cdot \standardMeasureS(S_{v,b}).
\end{equation}

As~the following theorem implies, every regular set of~trees $L$ has a~well\=/defined uniform measure $\standardMeasure{L}$.
\begin{thm}[\cite{michalewskiMeasureFinal}]
    Every regular language $L$ of~infinite trees is~universally measurable, i.e.~for every complete Borel measure $\mu$ on the set of~trees, we know that $L$ is~$\mu$\=/measurable.
\end{thm}

Hence, the following problem is~well\=/defined.

\begin{problem}[The $\standardMeasure{\text{MSO}}$ problem]
    \label{problem:computeMeasure}
   Is~there an algorithm that given an~MSO formula $\varphi$ computes $\standardMeasureBig{\lang(\varphi)}$?
\end{problem}

With the $\standardMeasure{\text{MSO}}$ problem we associate the following decision problem.

\begin{problem}[The positive $\standardMeasure{\text{MSO}}$ problem]
    \label{problem:positiveMeasure}
    Given an~MSO formula $\varphi$, decide whether $\standardMeasureBig{\lang(\varphi)} > 0$.
\end{problem}

If $\CC$ is~a class of~regular languages of~infinite trees, then by \emph{the (positive) $\standardMeasure{\CC}$ problem}, 
we understand the above where possible input languages are restricted to the class $\CC$.
If we restrict the class $\CC$ to formulae over the signature $\Sigma$ we denote it by $\CC(\Sigma)$.

The problem, in this form, was stated by Michalewski and Mio~\cite{michalewskiCompMeasure}.
It~is~open in the general case, but some partial results have been obtained, see the paragraph \emph{Related work} for details.

\section{Simple examples}
\label{sec:examples}

To~better understand the properties of~the uniform measures let us~consider some
simple sets of~infinite trees.
The presented examples not only give an~insight into the behaviour of~the uniform measures, but also will be used in~the proofs in~the following sections.

We start with a~simple lemma concerning the existence of sub\=/trees.

\begin{lemma}
    \label{lemma:basicMeasuresLemma}
    Let $\treeZ$ be~a~tree over the alphabet $\Gamma$ and~$u \in \setNodes$ be~a~position.
    \begin{enumerate}\label{enum:basiecMeasuresLemma}
        \item \label{item:finitePrefix} If $\treeZ$ is~finite and $L = \setBallInNode{\treeZ}{u} = \{\treeZ' \in \treesInfinite{{\Gamma}} \mid \isPrefix{\treeZ}{\subtreeAtNode{\treeZ'}{u}}\}$ then $\standardMeasure{L} = {\Gamma}^{-|\nodes{\treeZ}|}$.
        \item \label{item:finiteSubtree} If $\treeZ$ is~finite and 
        $L = \{\treeZ' \in \treesInfinite{{\Gamma}} \mid \exists v. (u \ancestor v) \land (\treeZ \weakAncestor \subtreeAtNode{\treeZ'}{v})\}$ then $\standardMeasure{L} = 1$.
        \item \label{item:infinitePrefix} If $\treeZ$ is~infinite and $L = \setBallInNode{\treeZ}{u} = \{\treeZ' \in \treesInfinite{{\Gamma}} \mid \isPrefix{\treeZ}{\subtreeAtNode{\treeZ'}{u}}\}$ then $\standardMeasure{L} = 0$.
    \end{enumerate}
\end{lemma}

\begin{proof}
    The proof of~Item~\ref{item:finitePrefix} is~straightforward. To~prove Item~\ref{item:finiteSubtree}, let~$L_i$ be~the set $L_i \eqdef \setBallInNode{\treeZ}{u\dL^i\dR} = \{\treeZ' \in \treesInfinite{{\Gamma}} \mid \isPrefix{\treeZ}{\subtreeAtNode{\treeZ'}{(u\dL^i\dR)}}\}$. Then, for every $j \geq 0$ we~have that~$L_j \subseteq L$ and, in~consequence, 
    $\overline{L} \subseteq \bigcap_{j \geq 0} \overline{L_j}$. Hence, for every $j \geq 0$ we~have that
    
    \[
    1 - \standardMeasure{L} = \standardMeasure{\overline{L}} \leq \standardMeasure{\bigcap_{j > i \geq 0} \overline{L_i}} = \big(1 - |\Gamma|^{-|\nodes{\treeZ}|}\big)^{j},
    \]
    where the last inequality follows from the fact that the nodes $u\dL^l\dR$ and $u\dL^k\dR$ are incomparable for $k \neq l$, thus $L_i$~are independent and
    \[
    \standardMeasure{\bigcap_{j > i \geq 0} \overline{L_i}} = \prod_{0 \leq i < j} \standardMeasure{\overline{L_i}}=  \prod_{0 \leq i < j} (1 -  |\Gamma|^{-|\nodes{\treeZ}|}) = \big(1 - |\Gamma|^{-|\nodes{\treeZ}|}\big)^{j}.
    \]
    Taking the limit, we~conclude Item~\ref{item:finiteSubtree}.
    
    To~prove Item~\ref{item:infinitePrefix}, let~$t_i$ be~a~sequence of~finite trees such that for every $i \geq 0$ we~have that $\treeZ_i \prefix \treeZ_{i+1} \prefix \treeZ$ and $|\nodes{\treeZ_i}| < |\nodes{\treeZ_{i+1}}|$.
    Since the sequence of~sets~$\setBallInNode{\treeZ_i}{u}$ is~decreasing and its limit contains the set $\setBallInNode{\treeZ}{u}$, i.e.~$\setBallInNode{\treeZ_i}{u} \supseteq \setBallInNode{\treeZ_{i+1}}{u} \supseteq \setBallInNode{\treeZ}{u}$, we~have that 
    $\standardMeasure{\setBallInNode{\treeZ}{u}} \leq \lim\limits_{i \to +\infty}\standardMeasure{\setBallInNode{\treeZ_i}{u}} = \lim\limits_{i \to +\infty} {|\Gamma|}^{-|\nodes{\treeZ_i}|} =  0$.
\end{proof}

The above examples may suggest that the uniform measures enjoy a~form of~\emph{Kolmogorov's zero\=/one law}:
e.g.~in~a~random tree a~given finite structure exists with probability~$1$, whereas a~given infinite structure exists with probability~$0$.
It~is not exactly the case, as~can be~seen by~the following examples. 

\begin{example}
    \label{ex:someMeasuresOfSets}
    Let $\Gamma = \{a,b,c\}$.
    \begin{enumerate}\label{enum:someMeasuresOfSets}
        \item \label{item:laprefixes} If~$L_a$ is~the set of~trees over the alphabet~$\Gamma$ with arbitrarily long sequences of~$a$\=/labelled nodes, i.e.
        \[
        L_a = \{\treeZ \in \treesInfinite{\Gamma} \mid \forall k \ge 0.\ \exists w,v \in \dS^{\ast}. \big(
        \left(|v| \geq k \right) \land \forall u \sqsubseteq v.\ \treeZ(wu) =a\big)\},
        \] then $\standardMeasure{L_a} = 1$.
        \item \label{item:la3path}If~$L_{a3}$ is~the language of~trees over the alphabet $\Gamma$ with an~infinite $\{a\}$\=/labelled path starting at~the root, i.e.
        \[
        L_{a3} = \{\treeZ \in \treesInfinite{\Gamma}\mid \exists w \in \dS^{\omega}.\  \forall u \sqsubsetneq w.\ \treeZ(u) = a\},
        \] then $\standardMeasure{L_{a3}}=0$.
        \item \label{item:la2path} If $L_{a2}$ is the language  of trees over the alphabet $\{a,b\}$ with an~infinite $\{a\}$\=/labelled path starting at~the root, i.e.
        \[
        L_{a2} = \{\treeZ \in \treesInfinite{\{a,b\}}\mid \exists w \in \dS^{\omega}.\ \forall u \sqsubsetneq w.\ \treeZ(u) = a\},
        \] then $\standardMeasure{L_{a2}}=0$.
        \item \label{item:labpath} If $L_{ab}$ is the language of trees over the alphabet $\Gamma$ with an~infinite $\{a,b\}$\=/labelled path starting at~the root, i.e. 
        \[
        L_{ab} = \{\treeZ \in \treesInfinite{\Gamma}\mid \exists w\in\dS^{\omega}.\  \forall u \sqsubsetneq w.\ \treeZ(u) \in \{a,b\}\},
        \] then $\standardMeasure{L_{ab}}=\frac{1}{2}$.

    \end{enumerate}
\end{example}

\begin{proof}[Calculating the measures]
    To~see Item~\ref{item:laprefixes}, let $\treeZ^i$ be~a~full tree of~height~$i$ such that every node in~$\nodes{\treeZ^i}$ is labelled~$a$
    and let $L^i$~be~the language of~trees having~$t^i$ as~a~prefix of~its~sub\=/tree at some node~$u$.
    Then, by~Item~\ref{item:finiteSubtree} of Lemma~\ref{lemma:basicMeasuresLemma}  we~have that $\standardMeasure{L^i} = 1$.
    Moreover, $\bigcap_{i \geq 1} L^i \subseteq L_{a}$ and $L^{i+1} \subseteq L^i$.
    Since every measure is~monotonically continuous, we have that 
    \[
    \standardMeasure{L_{a}} \geq \standardMeasureBig{\bigcap_{i \geq 1} L^i} = \lim\limits_{n \to +\infty} \standardMeasureBig{\bigcap_{n \geq i \geq 1} L^i} = 1.
    \]
    
    Let $\phi(\treeZ)$ stay for ``in~the tree~$\treeZ$ there is~an~infinite $\{a\}$\=/labelled path starting at~the root'' then
    Item~\ref{item:la3path} follows from the fact that the language in~question is~regular, thus measurable, and its measure satisfies the following equation.%
    \footnote{Here, we~slightly abuse the notation for the sake of~readability.
        We write $\standardMeasure{\psi}$ instead of~$\standardMeasure{ \{t \in \treesInfinite{\Gamma} \mid \psi\}}$. Moreover, we write $\psi(t.u)$ for the statement ``the sub\=/tree of $t$ in $u$ satisfies $\psi$''.}
    \begin{align*}
    \standardMeasure{L_{a3}} = \standardMeasureBig{\varphi(t) \land \treeZ(\varepsilon) {\not=} a}  & +\; \\
    \standardMeasureBig{\treeZ(\varepsilon) {=} a}&\cdot
    \big(
    \standardMeasure{\phi(\subtreeAtNode{\treeZ}{\dL})}
    +
    \standardMeasure{\phi(\subtreeAtNode{\treeZ}{\dR})}
    - \standardMeasure{\phi(\subtreeAtNode{\treeZ}{\dL}) \land \phi(\subtreeAtNode{\treeZ}{\dR})}
    \big)
    \end{align*}
    The equation states that the measure of~$L_{a3}$ is~equal to~the sum of the measures
    of~two sets of~trees.
    The first set consists of~all trees~$t$ such that the root is~not labelled~$a$ and the tree $t$~satisfies $\phi$.
    The second set consist of~all trees $t$ such that the root is~labelled~$a$, the sub\=/tree at~the left child 
    of~the root satisfies~$\phi$, i.e.~$\phi(\subtreeAtNode{\treeZ}{\dL})$, or~the sub\=/tree at~the right child 
    of~the root satisfies~$\phi$, i.e.~$\phi(\subtreeAtNode{\treeZ}{\dR})$.
    
    Since the set of~all possible trees at~left child (or, at~right child) of~the root is~the set of~all trees,
    we~get the equation
    \[\standardMeasure{L_{a3}} = \frac{1}{3}\cdot \big(2\standardMeasure{L_{a3}} - \standardMeasure{L_{a3}}^2\big)=\frac{2}{3}\standardMeasure{L_{a3}}-\frac{1}{3}\cdot\standardMeasure{L_{a3}}^2\]
    implying that $\standardMeasure{L_{a3}} =  0$ or $\standardMeasure{L_{a3}} = -1$.
    Since the measure cannot be~negative, we~conclude that $\standardMeasure{L_{a3}} =  0$.
    
    Similarly, in~Item~\ref{item:la2path} we~get the equation
    \begin{equation}
    \label{example:measure4}
    \standardMeasure{L_{a2}} = \frac{1}{2}\cdot \big(2\standardMeasure{L_{a2}} - \standardMeasure{L_{a2}}^2\big)=\standardMeasure{L_{a2}}-\frac{1}{2}\cdot\standardMeasure{L_{a2}}^2
    \end{equation}
    implying that $\standardMeasure{L_{a2}} =  0$.
    
    In~Item~\ref{item:labpath}, we~get the equation
    \begin{equation}
    \label{example:measure3}
    \standardMeasure{L_{ab}} = \frac{2}{3}\cdot \big(2\standardMeasure{L_{ab}} - \standardMeasure{L_{ab}}^2\big)=\frac{4}{3}\standardMeasure{L_{ab}}-\frac{2}{3}\cdot\standardMeasure{L_{ab}}^2
    \end{equation}
    implying that either $\standardMeasure{L_{ab}} =  \frac{1}{2}$ or~$\standardMeasure{L_{ab}} =  0$.
    Thus we~need to~look at~this example a~bit more carefully.
    Consider a~sequence of~languages~$\{A^i\}_{i \geq 0}$, where $A^0 = \treesInfinite{\Gamma}$
    and $A^i$~is~the language such that there is~an~$\{a,b\}$-labelled path of~length~$i$ beginning at~the root.
    Then, we~claim the following.
    
    \begin{claim}
        \label{claim:prefixesConvergeExamplesMeasures}
        $\bigcap_{i \geq 1} A^i = L_{ab}$
    \end{claim}
    \begin{proof}
        Since an infinite path contains sub\=/paths of~arbitrary length, we~have that $\bigcap_{i \geq 1} A^i \supseteq L_{ab}$.
        For the reverse inclusion, let $t \in \bigcap_{i \geq 1} A^i$.
        Then, for every $i \geq 0$ the tree~$t$ has an~$\{a,b\}$\=/labelled path of~length~$i$.
        Since~$t$ is~a~binary tree, K\"{o}nig's lemma assures that the tree~$t$ has an~infinite $\{a,b\}$\=/labelled path. This concludes the proof of~the claim.
    \end{proof}
    
    Now, for every $i \geq 0$ we~have that $A^{i+1} \subseteq A^i$ and 
    \begin{equation}
    \label{example:recursiveMeasure}
    \standardMeasure{A^{i+1}} = \frac{2}{3}\cdot \big(2\standardMeasure{A^i} - \standardMeasure{A^i}^2\big)=\frac{4}{3}\standardMeasure{A^i}-
    \frac{2}{3}\cdot\standardMeasure{A^i}^2.
    \end{equation}
    Note that if~$\standardMeasure{A^i} \ge \frac{1}{2}$, then $\standardMeasure{A^{i+1}} \ge \frac{1}{2}$. Indeed, the quadratic function $f(x) = \frac{2}{3}(2x - x^2)$ is monotonically  increasing on the interval $[-\infty, 1] $ and we have that $f(1) = \frac{2}{3}$ and $f(\frac{1}{2}) = \frac{1}{2}$. Since $\standardMeasure{A^0} = 1 \ge \frac{1}{2}$, we conclude that $\standardMeasure{L_{ab}} =  \frac{1}{2}$. 
\end{proof}

Before we proceed, observe that the languages $L_{a3}$, $L_{a2}$, and $L_{ab}$ can all be recognised by some non\=/deterministic safety automata (see Proposition~\ref{pro:closed-is-safety} and Section~\ref{sec:safety-auto}). Thus, instead of computing the measures by hand, we could have invoked Theorem~\ref{thm:openSetsUpperBound}. It is not incidental, because the idea of~inductive approximation of the~measure of~the set $L_{ab}$, expressed by~Equation~\eqref{example:recursiveMeasure}, is the cornerstone of the general construction performed in~Section~\ref{sec:open-sets}.

\section{First-order definable languages}
\label{sec:fo-sets}

The ideas presented in~both Lemma~\ref{lemma:basicMeasuresLemma} and Example~\ref{ex:someMeasuresOfSets} allow us~to~compute the measures of~sets of trees defined by~certain first\=/order formulae.

\subsection{First-order definable languages without descendant}
\label{subsec:fo-sets-no-desc}

\begin{thm}
    \label{thm:FOcomputable}
    Let $\Gamma$ be an~alphabet and let $\varphi$ be~a~first\=/order sentence over the signature~$\{\rootP, s_\dL, s_\dR, s\}\cup \Gamma$. Then the measure $\standardMeasure{\lang(\varphi)}$ is~rational and computable in~three\=/fold exponential space.
\end{thm}
The proof utilises the \emph{Gaifman locality} to~partition the formula into two separate sub\=/formulae.
Intuitively, one sub\=/formulae describes the neighbourhood of~the root while the remaining one describes the tree ``far away from the root''.

\paragraph*{Gaifman normal form}
Let $\AA$ be a~relational structure.
 The \emph{Gaifman graph} of~$\AA$ is the undirected graph $G^{\AA}$ where the set 
of vertices is the universe of~$\AA$ and there is an edge between two vertices in $G^{\AA}$ if there is a relation $R$ in $\AA$ 
and a tuple $x \in R$ that contains $u$ and $v$.
The Gaifman distance $\td(u,v)$ between two elements $u$, $v$ of the universe of $\AA$ is the distance between $u$ and $v$ in the Gaifman graph.

Notice that if~we~exclude the ancestor relation from the tree structure, then the Gaifman graph of~a~tree $t$ is~induced by~the child relations only and the Gaifman distance coincides with the distance between the positions in the tree. This means, in particular, that for a~fixed finite distance $l$ and any~given node $u \in \dom{t}$ there is~only finitely many positions $v \in \dom{t}$ such that the Gaifman~distance between $u$ and $v$ is $l$ or less. However, if we~allow the ancestor relation then any two nodes in the tree are in Gaifman distance two or~less.
In~this section, from now on, we~exclude $\ancestor$ from the signature.

For a~natural number $r\in\bbN$, let $[\td(x,y) \leq r]$ be a~first\=/order formula stating that the distance between $x$ and $y$ is at most $r$. This formula has two free variables $x$, $y$ and its size depends on~$r$. Similarly, the negation of that formula will be denoted $[\td(x,y)>r]$.

We say that a first\=/order formula $\varphi(x)$ is an~\emph{$r$\=/local formula around $x$} if the quantifiers of $\varphi$ are restricted to the
$r$\=/neighbourhood of $x$, i.e.~if the quantifiers inside $\varphi(x)$ (except those inside the formulae $[\td(x,y) < r]$) have the form $\forall^{\le r}$ or $\exists^{\le r}$ defined as follows:
\begin{align*}
\exists^{\le r} y.\ \psi(y) &\eqdef \exists y.\ [\td(x,y) \leq r] \land \psi(y)\\
\forall^{\le r} y.\ \psi(y) &\eqdef \forall y.\ [\td(x,y) \leq r] \to \psi(y).
\end{align*}

We say that a~first\=/order sentence $\varphi$ is a \emph{basic $r$-local sentence} if it is of the form
\begin{equation}
\label{eq:basicLocalFormula}
\varphi \eqdef \exists x_1,\dots,x_n \left( \bigwedge_{i=1}^{n} \varphi^r_i(x_i) \land \bigwedge_{1 \leq i<j \leq n} [\td(x_i,x_j) > 2r] \right),
\end{equation}
where $\varphi_i^r(x)$ are $r$\=/local formulae around $x$.

\begin{thm}[Gaifman]
    \label{thm:gaifman}
    
    Every first-order sentence is equivalent to a Boolean combination of basic $r$-local sentences, where $r$ is a number depending
    on the size of the formula.
    Furthermore, $r$ can be chosen so that $r \leq 7^{\textit{qr}(\varphi)}$, where $\textit{qr}(\varphi)$ is the quantifier rank of $\varphi$.
\end{thm}

As proved by Heimberg et al., cf.~\cite{heimbergOptimalGaifman}, the translation to Gaifman normal form can be costly.

\begin{thm}[\cite{heimbergOptimalGaifman}]
    \label{thm:gaifmanCost}
    There is a three-fold exponential algorithm on structures of degree 3 that transforms a first-order formula into its Gaifman normal form.
    Moreover, there are first-order formulae for which the three-fold exponential blow-up is unavoidable.
\end{thm}

\paragraph*{Root formula}
Now we~define the idea of~a~\emph{root formula}, i.e.~a~formula that necessarily describes the neighbourhood of the root. 
Let $\psi(x)$ be~an~$r$\=/local formula around~$x$.
We~say that $\psi(x)$~is~a~\emph{root formula} if~for
every tree~$\treeZ \in \treesInfinite{{\Gamma}}$ and every position~$u \in \dS^{\ast}$ 
if~$\treeZ, u \models \psi(x)$ then $\td(u,\varepsilon) < r $.
Note that every unsatisfiable formula is, by~the definition, a~root formula.

Let $\varphi$ be~a~basic~$r$\=/local sentence, i.e.~of~the form given by~\eqref{eq:basicLocalFormula}.
We~say that~$\varphi_i^r$, for $i \in \{1,\dots,n\}$, is~a~\emph{root} formula of~$\varphi$ if $\varphi_i^r$ is~a~root formula.

\begin{fact}
    \label{fact:uniqueRootFormula}
    For every satisfiable basic $r$\=/local sentence there is~at~most one root formula.
\end{fact}

In~other words, only one of~$\varphi_i^r$s can describe the $r$\=/neighbourhood of~the root.
As, if~two such formulas $\varphi_i^r$, $\varphi_j^r$ would be~root formulae, then the variables $x_i$ and $x_j$
would be~mapped in~a~distance at~most $r{-}1$ from the root, i.e.~in~a~distance strictly smaller than $2r$ from each other.

Note that, by~the definition of~satisfiability, for a~tree $\treeZ \in \treesInfinite{\Gamma}$ and a~basic $r$\=/local sentence $\varphi$ we~have that
$\treeZ \models \varphi$ if~and only if~there is~a~function $\fun{\tau}{\{x_1,\dots,x_n\}}{\dS^{\ast}}$
mapping variables~$x_1,\dots, x_n$ to~nodes of $t$ so~that for every $i \in \{1,\dots,n\}$
we~have that $\treeZ, \tau(x_i) \models \varphi_i(x_i)$ and for every pair of~indices $i\neq j$ we~have that
$\td(\tau(x_i), \tau(x_j)) > 2r$.

\begin{lemma}
    \label{lemma:basicSentenceMeasure}
    Let $\varphi$ be~a~basic $r$\=/local sentence, i.e.~as~in~\eqref{eq:basicLocalFormula}.
    If~$\varphi$ is 
    \begin{itemize}\itemsep=0pt
        \item not satisfiable, then $\standardMeasure{\lang(\varphi)} = 0$,
        \item satisfiable and has no~root formula, then $\standardMeasure{\lang(\varphi)} = 1$,
        \item { satisfiable and has a~root formula~$\varphi^{\ast}$,\\ then for every~$\tau$ that is~a~full tree of~height~$2r{+}1$ we have that
            \[
            \standardMeasure{\lang(\varphi) \cap \bB_{\tau}} =
            \begin{cases}
            \standardMeasure{\bB_{\tau}} & \quad \text{if } \exists u\in\dS^{< r}.\ \tau, u \models  \varphi^{\ast}(x);\\
            0 & \quad \text{otherwise}.
            \end{cases}
            \]
        }
    \end{itemize}
\end{lemma}
\begin{figure}
\centering

\begin{tikzpicture}[scale=0.62]

	\newcommand{\smalltree}[2]{
		\draw (#1.center)+(0,-.2) -- +(-1.5,-2)
		(#1.center)+(0,-.2) -- +(1.5,-2);
      	\draw[dotted] (#1) + (-1,0) -- ++ (1,0);
        \draw[dotted] (#1) + (-1,-1) -- ++(1,-1);
        
        \draw[<->] (#1) + (-1,0) -- ++(-1,-1) node[midway,left] {$r$}; 
       	\coordinate (cc) at ($(#1)+(0.0,-1.9)$);
       	\draw[dashed] (cc) circle(1.0);
       	\draw (cc) node[fill, minimum size=3pt, inner sep=0, circle] {};
        \draw (cc)++(0.3,-0.2) node[scale=0.7] {$u_{#2}$};
        \draw[<->] (cc) -- ++(-0.707,-0.707) node[midway,left] {$r$}; 
	}

    \node (root) at (0,0) {$ $};
    \node (lend) at (-10.5,-9) {$ $};
    \node (rend) at (+10.5,-9) {$ $};
    
    \node (x1) at (0.5,-2) {$ $};
    \node (x2) at (-1.5,-3) {$ $};
    \node (x3) at (1.5, -4) {$ $};
    \node (xn) at (4.5, -5) {$ $};
    \node (v1) at (-3, -4) {$v_1$};
    \node (v2) at (-1, -4) {$v_2$};
    \node (v3) at (1.5, -5) {$v_3$};
    \node (vn) at (5.5,-6) {$v_n$};
	\node [label={[shift={(.5,-.7)}]$t_1$}] (t1) at (-4.5,-5) {}; 
	\node [label={[shift={(.5,-.7)}]$t_2$}] (t2) at (-1.2, -5) {};
	\node [label={[shift={(.5,-.7)}]$t_3$}] (t3) at (2.0, -6) {};
	\node [label={[shift={(.5,-.7)}]$t_n$}] (tn) at (5.5, -7) {};
	
	\draw (root.center) -- (lend);
	\draw (root.center) -- (rend);
	\draw[dashed, thick]
		(root) -- (x1)
		(root) -- (xn)
		(x1) -- (x2)
		(x2) -- (v1)
		(x2) -- (v2)
		(x1) -- (x3)
		(x3) -- (v3)
		(x3) -- +(1,-1)
		(xn) -- (vn)
		(xn) -- +(-1,-1)
		(v1) -- (t1.center)
		(v2) -- (t2.center)
		(v3) -- (t3.center)
		(vn) -- (tn.center);
	\smalltree{t1}{1}
	\smalltree{t2}{2}
	\smalltree{t3}{3}
	\smalltree{tn}{n}
	
	\draw[dotted] (-10,0) -- (10,0);
	\draw[dotted] (-10,-3.5) -- (10,-3.5);
	\draw[<->] (-10,0) -- (-10, -3.5) node[midway,left] {$2r$}; 
\end{tikzpicture}
\caption{The tree from family $F$ in the proof of Lemma~\ref{lemma:basicSentenceMeasure}.}

\comment{
\begin{tikzpicture}
	\node [label={[shift={(.5,-.7)}]$t_i$}] (root) at (0,0) {};
	\node (v) at (0,-2.5) {$u_i$};
	
	\draw 
		(root.center) -- (-4,-4)
		(root.center) -- (4,-4);	
		
	\draw[dotted]
		(-4,0) -- (4,0)
		(-4,-2) -- (4,-2);
	\draw[<->] (-4,0) -- (-4, -2) node[midway,left] {$r$};
	\draw[dashed] (v) circle(1.5);
\end{tikzpicture}
}
\label{fig:embedding_FO_subtrees}
\end{figure}
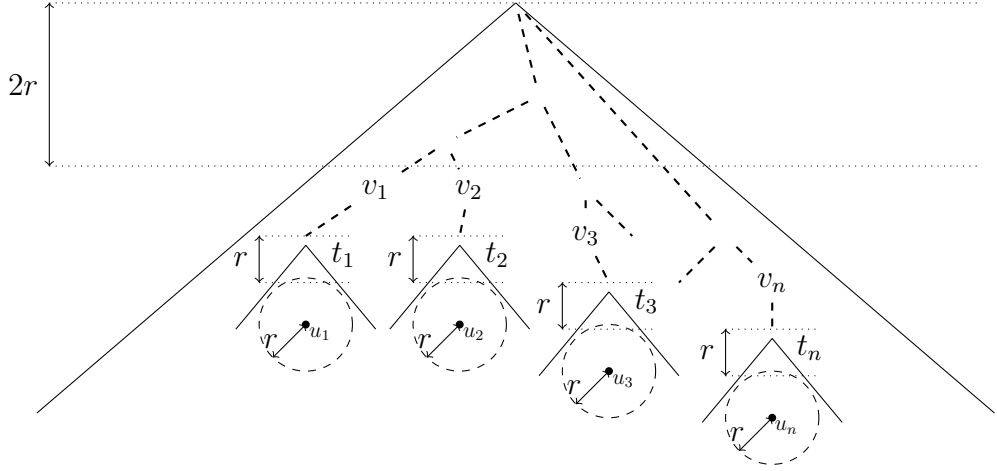
\begin{proof}
    If $\varphi$ is~not satisfiable then $\lang(\varphi) = \emptyset$ and $\standardMeasure{\lang(\varphi)} = 0$.
    Therefore, let us~assume that $\varphi$ is~satisfiable.
    By~Fact~\ref{fact:uniqueRootFormula} we~know that there is~at~most one root formula in~$\varphi$.
    Let $I$~be~the set of~indices of~non\=/root formulae, i.e.~for
    $i \in I$ we~have that $\varphi_i$ is~not a~root formula. Since $\varphi$ is satisfiable, 
    for every $i \in I$ there are~a~finite tree $\treeZ_i \in  \treesFinite{{\Gamma}}$ and a~node $u_i\in\dS^\ast$ of~length $|u_i|>r$, such that $\treeZ_i,u_i \models \varphi_i(x)$ and the set
    $\nodes{\treeZ_i}$ contains the $r$\=/neighbourhood of~$u_i$.
    
    Let~$W = \{v_i\}_{i=1}^n$ be~a~set of~$n$ $\sqsubseteq$\=/incomparable nodes such that for $i \in I$ we~have that~$|v_i| > 2r$.
    Let $F = \bigcap_{i \in I} L_i$ where~$L_i$ is~the set of~trees for which $\treeZ_i$ is~a~prefix of~a~sub\=/tree at~some node below~$v_i$, i.e.~$L_i \eqdef \{\treeZ' \in \treesInfinite{{\Gamma}} \mid \exists u. (v_i \ancestor u) \land (\treeZ_i \prefix \subtreeAtNode{\treeZ'}{u})\}$.
    By~Lemma~\ref{lemma:basicMeasuresLemma} every~$L_i$ has measure~$1$, thus we~have that $\standardMeasure{F}=1$.
    Moreover, for every tree~$\treeZ \in F$ and index~$i\in I$  there is~a~node~$v_i' \in \dS^\ast$ such that $\td(v_i,v'_i)> r$, $v_i \sqsubseteq v_i'$ and $\treeZ,v_i' \models \varphi_i(x)$.
    
    Now, if~there is~no~root formula in~$\varphi$, i.e.~$I = \{1\dots,n\}$, then $F \subseteq \lang(\varphi)$.
    Indeed, let $\treeZ \in F$, then for $i\neq j$ we have that $\td(v_i',v_j')>2r$ and we~can infer that $\treeZ \models \varphi$. Hence, the sequence of~inequalities 
    \[1 = \standardMeasure{F}  \leq \standardMeasure{\lang(\varphi)} \le 1\]
    is~sound and proves the second bullet of~Lemma~\ref{lemma:basicSentenceMeasure}.
    
    
    Consider the opposite case that there is a~root formula in~$\varphi$. Without loss of generality, $\varphi_1$ is~the root formula and $I = \{2,\dots,n\}$. Moreover, let $F$ and $v_i'$s be~as~before, let $\tau$ be~a~full tree of~height~$2r{+}1$, as~stated in~the lemma, and $\treeZ \in F \cap \mathbb{B}_{\tau}$ be~a~full tree.	
    
    If~there is~$u_1\in\dS^{< r}$ such that $\tau, u_1 \models  \varphi_1(x_1)$, then we~take $v_1' \eqdef u_1$.
    Now, again, for $i\neq j$ we~have that $\td(v_i',v_j')>2r$ and for all $i \in I$ we~have that $\treeZ, v_i' \models \varphi_i(x_i)$. 
    In~other words, if~there is~$u_1\in\dS^{<r}$ such that $\tau, u_1 \models  \varphi_1(x)$ then
    $F \cap \mathbb{B}_{\tau} \subseteq \lang(\varphi) \cap \mathbb{B}_{\tau}$.
    Moreover, since $F$ is~of~measure~$1$, the following sequence of~inequalities is~sound
    \[\standardMeasure{\mathbb{B}_{\tau}} = \standardMeasure{F \cap \mathbb{B}_{\tau}}  \leq \standardMeasure{\lang(\varphi) \cap \mathbb{B}_{\tau}} \leq \standardMeasure{\mathbb{B}_{\tau}}.\]
    
    Now assume that there is~no such~$u_1$. We claim that in that case $\lang(\varphi_1)\cap \mathbb{B}_{\tau}=\emptyset$. Indeed, if there were a~tree $t$ in~the intersection, then the definition of the root formula would provide a~node $u_1\in\{\dL,\dR\}^{<r}$ making $t,u_1\models \varphi_1(x_1)$ true. But it would mean that $\tau,u_1\models \varphi_1(x_1)$ as~well by the form of the quantifiers inside $\varphi_1$. Thus, in that case $\standardMeasure{\lang(\varphi) \cap \mathbb{B}_{\tau}} = 0$, which concludes the proof.
\end{proof}

Intuitively, the above lemma states that, when we~consider the uniform measure and a~basic $r$\=/local sentence, the behaviour of~the sentence is~almost surely defined by~the neighbourhood of~the root. 
This intuition can be~formalised as~follows.

\begin{lemma}
    \label{lemma:reductionOfAFormula}
    Let $\varphi$ be~a~basic $r$\=/local sentence.
    Then there is~a~sentence~$\varphi^{\ast}$ such that for every full tree~$\tau$ of~height $2r+1$
    we~have that
    \[ \standardMeasure{\lang(\varphi) \cap \bB_{\tau}} =  \standardMeasure{\lang(\varphi^{\ast}) \cap \bB_{\tau}}.\]
    Moreover, for every full tree~$\treeZ \in \bB_{\tau}$ we~have that $\treeZ \models \varphi^{\ast}$ if~and only if~$\tau \models \varphi^\ast$.
\end{lemma}

\begin{proof}
    If~$\varphi$ has a~root formula~$\varphi_i$, then we~take $\varphi^{\ast} \eqdef \exists x.\ \varphi_i(x) \land [\td(x,\varepsilon)<r]$.
    If~$\varphi$ has no~root formulae but is~satisfiable then we~take $\varphi^{\ast} \eqdef \exists x.\ \rootP(x)$.
    Otherwise, we~take $\varphi^{\ast}\eqdef\bot$.
\end{proof}

The formula~$\varphi^{\ast}$ is~called the \emph{reduction} of~$\varphi$.
Before we~show how to~compute the reduction of~a~basic $r$\=/local formula, we~recall 
a~known result.

\begin{lemma}[Folklore]
    \label{lemma:modelCheckingFO}
    There is~an~algorithm that given a~first\=/order sentence~$\varphi$ and a~finite~tree~$\tau$ 
    decides whether $\tau \models \varphi$ in~space polynomial with respect to~the size of~the formula~$\varphi$
    and with respect to~the size of~the set of~nodes of~the tree~$\tau$.
\end{lemma}

\begin{proof}
    The lemma is~folklore; the property can be~easily verified using an~alternating polynomial time (\mbox{\aptime}) algorithm.
\end{proof}

Now we~show how to~compute the reduction.

\begin{lemma}
    \label{lemma:computingReduction}
    Given a~basic $r$\=/local sentence $\varphi$ one can compute its reduction~$\varphi^*$
    in~space polynomial in~the size of~the formula and doubly exponential in~the unary encoding of~$r$.
\end{lemma}

\begin{proof}
    An~$r$\=/local formula $\psi(x)$ is~satisfiable in~some~full tree if, and only if, it~is~satisfiable in~a~node~$u$ of~some~tree of~height~$2r+1$, such that $|u| < r+1$.
    Thus, to~check the satisfiability of~any formula $\varphi_i$, we~need to~check the trees of~height at~most $2r+1$.
    
    Moreover, to~check whether $\varphi_i$ is~a~root formula, we need to check whether $\varphi_i$ is~satisfiable
    and the formula $\varphi_i(x) \land [\td(x,\varepsilon)\geq r]$
    is~not satisfied in~any full tree.
    This, again, can be~checked by~iterating over all trees of~height at~most~$2r+3$.
    
    Thus, the reduction of~$\varphi$ can be~computed by the algorithm \textsc{computeReduction}
    presented in~Algorithm~\ref{alg:computeReduction}.
    The complexity follows from Lemma~\ref{lemma:modelCheckingFO}.
\end{proof}

\begin{algorithm}                      
    \caption{\textsc{computeReduction}}          
    \label{alg:computeReduction}                           
    \begin{algorithmic}                    
        \REQUIRE a first\=/order sentence $\varphi$ in~Gaifman normal form 
        \STATE $S \gets \{ i \mid \varphi_i \text{ is not satisfiable} \}$
        \IF {$|S| > 0$}
        \RETURN $\bot$
        \ENDIF
        \STATE $S \gets \{ i \mid \varphi_i \text{ is~a~root formula} \}$
        \IF {$|S| = 0$}
        \RETURN $\exists x.\ \rootP(x)$ 
        \ELSIF {$|S| = 1$}
        \STATE $i \gets S.any()$
        \RETURN $\exists x.\ \varphi_i(x) \land [\td(x,\varepsilon)<r]$ 
        \ELSE
        \RETURN $\bot$
        \ENDIF
    \end{algorithmic}
\end{algorithm}

Lemma~\ref{lemma:reductionOfAFormula} can be~extended to~Boolean combinations of~basic $r$\=/local sentences by~the following property of~measurable sets.

\begin{lemma}
    \label{lemma:booleanAlgebraAndMeasure}
    Let $M$ be~a~measurable space with measure~$\mu$, $W$ be~a~$\mu$\=/measurable set, and $\{ S_i \}_{i \in I}$ be~a~family of~$\mu$\=/measurable sets such that for every $i \in I$ either $\mu(W \cap S_i) = 0$ or~$\mu(W \cap S_i) = \mu(W)$.
    Then, for every set~$S$ in the Boolean algebra of~sets generated by~$\{ S_i \}_{i \in I}$, we have that
    either $\mu(W \cap S) = 0$ or $\mu(W \cap S) = \mu(W)$.
\end{lemma}

\begin{proof}
    The proof goes by~a~standard inductive argument.
\end{proof}

Hence, by~Lemma~\ref{lemma:basicSentenceMeasure} and the above lemma, we~obtain the following.

\begin{lemma}
    \label{lemma:standardMeasureOfBCOfBasicFormulae}
    Let $\phi$ be~a~Boolean combination of~basic $r$\=/local sentences and $\tau$~be~a~full tree of~height~$2r{+}1$.
    Then, $\standardMeasure{ \lang(\phi) \cap  \bB_{\tau}} = \standardMeasure{ \lang(\phi^{\ast}) \cap  \bB_{\tau}}$, where $\phi^*$ is the reduction of~$\phi$, i.e.~the Boolean combination~$\phi$ with its every basic $r$\=/local sentence~$\varphi$ replaced by~its reduction~$\varphi^*$.
    
    Moreover,
    \[\standardMeasure{ \lang(\phi^{\ast}) \cap  \bB_{\tau}}=
    \begin{cases}
    \standardMeasure{\bB_{\tau}} & \quad \text{if } \tau \models  \phi^{\ast};\\
    0 & \quad \text{otherwise.}
    \end{cases}
    \]
\end{lemma}

\begin{proof}
    For every full tree $\tau$ of~height~$r$ we~use Lemma~\ref{lemma:booleanAlgebraAndMeasure} with $M$~being the set of~full trees~$\treesInfinite{\Gamma}$,
    $\mu$ being the uniform measure~$\standardMeasureS$, and~$W$ being the set $\mathbb{B}_{\tau}$. The sets $S_i$ are the sets
    of~trees defined by~the basic $r$\=/local sentences and~$S=\lang(\phi)$.
    By~Lemma~\ref{lemma:reductionOfAFormula}, the assumptions of~Lemma~\ref{lemma:booleanAlgebraAndMeasure}
    are satisfied.
\end{proof}

With the above lemmas, we~can finally prove Theorem~\ref{thm:FOcomputable}.

\begin{proof}[Proof of Theorem~\ref{thm:FOcomputable}]
    
    Let $\varphi$ be~a~first\=/order sentence as~in~the theorem.
    We~utilise the Gaifman locality theorem~(see Theorem~\ref{thm:gaifman} on page~\pageref{thm:gaifman}) to~translate the sentence~$\varphi$ into a~Boolean combination~$\phi$ of~basic $r$\=/local sentences.
    Now, let $\phi^{\ast}$~be~the reduction of~$\phi$, as~in~Lemma~\ref{lemma:standardMeasureOfBCOfBasicFormulae},
    and let $S = \trees{{\Gamma}}^{2r+1}$ be~the set of~all full trees of~height~$h = 2r{+}1$. Then, 
    \[
    \begin{array}{r c l c l }
    \standardMeasure{\lang(\phi)}& \eqext{1} & \standardMeasure{\lang(\phi) \cap \big(\bigcup_{\tau \in S} \bB_{\tau}\big)} & \eqext{2} &  \standardMeasure{\bigcup_{\tau \in S} \big(\lang(\phi) \cap  \bB_{\tau}\big)}\\ 
    &\eqext{3}& \sum_{\tau \in S} \standardMeasure{ \lang(\phi) \cap  \bB_{\tau}} & \eqext{4}  & 
    \sum_{\tau \in S} \standardMeasure{ \lang(\phi^{\ast}) \cap  \bB_{\tau}}\\ 
    & \eqext{5}  & \sum_{\tau \in S \land \tau \models \phi^\ast} \standardMeasure{\bB_{\tau}} & \eqext{6} & |\{\tau \in S \mid \tau \models \phi^\ast\}|\cdot\frac{1}{|\Gamma|^{2^{h+1}-1}}.
    \end{array} 
    \]
    The first equation follows from the fact that the sets in $\{\bB_{\tau} \mid \tau \in S\}$ are pairwise disjoint. The second from operations on~sets and the third is~a~simple property of~measures.
    The fourth follows from the first part of~Lemma~\ref{lemma:standardMeasureOfBCOfBasicFormulae},
    while the fifth follows from the second part of~this lemma.
    The last equation is~a~consequence of~the fact that~$\standardMeasure{\bB_{\tau}} = {|\Gamma|}^{-|\dom{\tau}|}$.
    
    Since $\standardMeasure{\lang(\phi)} = \frac{|\{\tau \in S \mid \tau \models \psi\}|}{{|\Gamma|}^{2^{h+1}-1}}$, it is enough to count how many full trees of height $h = 2r{+}1$ satisfy the reduction of $\phi$.
    The pseudo\=/code of~the algorithm, called \textsc{computeMeasureFO}, is~presented in~Algorithm~\ref{alg:FOMeasure}.
    
    The complexity upper bound comes from the fact that translating a~first\=/order sentence $\varphi$ into its~Gaifman normal form can be done in~three\=/fold exponential time and can produce a~three\=/fold exponential sentence~$\phi$ in~result, see~\cite{modelTheoryMakesFormulasLarge} for details.
    The resulting sentence~$\phi$ is~a~Boolean combination of~basic $r$\=/local sentences, 
    thus, we~can compute its reduction in~three\=/fold exponential space.
    The function $\textsc{computeReduction}^{\ast}$ computes the reduction $\phi^\ast$ of~the Boolean combinations
    by~replacing the sentences used in~the Boolean combination with their reductions.
    This can be~done in~the required complexity, see~Lemma~\ref{lemma:computingReduction} and note that the size of~the sentence dominates the constant~$r$.
    Finally, the last part of~the algorithm requires us~to~check the sentence~$\phi^\ast$ against three\=/fold exponential number of~trees of~size that is~two\=/fold exponential in~the size of~the original formula.
    Since model checking of~a~first\=/order sentence can be~done in~polynomial space with respect to~the size of~the tree and to~the size of~the sentence, see Lemma~\ref{lemma:modelCheckingFO}, we~get the upper bound.
\end{proof}

\begin{algorithm}                      
    \caption{\textsc{computeMeasureFO}}          
    \label{alg:FOMeasure}                           
    \begin{algorithmic}                    
        \REQUIRE a~first\=/order sentence $\varphi$ and a~positive number~$h$
        \STATE $S \gets \text{the set of~all full trees of~height } h$ 
        \STATE $\phi \gets \textsc{computeGaifmanForm}(\varphi)$ 
        \STATE $\phi \gets \textsc{computeReduction}^{\ast}(\phi)$
        \STATE $S \gets \{\treeZ \in S \mid  \treeZ \models \phi\} $
        
        \RETURN ${|S|}\cdot {{|\Gamma|}^{-2^{h+1}+1}}$
    \end{algorithmic}
\end{algorithm}

The following remark follows directly from the construction.

\begin{remark}
\label{rem:fo-as-clopen}
The above theorem implies that $\standardMeasureBig{\lang(\varphi)}=\standardMeasureBig{\lang(\varphi^\ast)}$, where $\varphi^\ast$ is the reduction of the given formula $\varphi$. Moreover, Lemma~\ref{lemma:reductionOfAFormula} implies that $\lang(\varphi^\ast)$ is a~clopen set because it is a~finite union of basic sets.
\end{remark}

\subsection{First-order definable languages with descendant}
\label{subsec:fo-sets-with-descendant}

The technique used to~prove Theorem~\ref{thm:FOcomputable} cannot be~extended to~formulae utilising the descendant 
relation because when we~allow the descendant relation, the diameter of~the Gaifman graph of~any tree is~at most two.
Additionally, as~presented in~Proposition~\ref{prop:algebraicValueFO} below, sets of~full trees defined by~such formulae can have irrational measures.

\begin{proposition}
    \label{prop:algebraicValueFO}
    There is~a~set~of~full trees over an~alphabet~$\Gamma$ that is definable by~a~first\=/order formula over the signature $\{\rootP, s_{\dL}, s_{\dR}, s, {\ancestor} \}\cup \Gamma$ and the uniform measure of this set is~irrational.
\end{proposition}

\begin{proof}
    Let $\Gamma = \{a,b\}$, we~define a~language~$L$ in~the following way $L \eqdef \{\treeZ \in \treesInfinite{\{a,b\}} \mid$ for every path the earliest node labelled~$b$
    (if~exists) is~at~an~even depth$\}$. We will prove that the measure $\standardMeasure{L}$ is~irrational, and 
    there is~a~language $L'$ definable by~a~first\=/order formula over the signature
    $\{ s_{\dL}, s_{\dR}, {\ancestor} \}\cup \Gamma$ such that $\standardMeasure{L'} = \standardMeasure{L}$.
    We~start by~computing the measure of~$L$, then we~will define~$L'$.
    
    Observe that the measure $\standardMeasure{L}$ satisfies the following equation.
    \[
    \standardMeasure{L} = \standardMeasureBig{\{\treeZ \in \treesInfinite{\{a,b\}} \mid \treeZ(\varepsilon) {=} b  \}} +
    \standardMeasureBig{\{\treeZ \in \treesInfinite{\{a,b\}} \mid \treeZ(\varepsilon){=}\treeZ(\dL){=}\treeZ(\dR){=} a\} } \cdot \standardMeasure{L}^4
    \]
    The equation says that the trees in $L$ either contain $b$ at the root (i.e.~$t(\varepsilon)=b$) or contain $a$ in the first three vertices: $\varepsilon$, $\dL$, and $\dR$; and the four subtrees of $t$ under all nodes of length $2$ (i.e.~$t.\dL\dL$, $t.\dL\dR$, $t.\dR\dL$, and $t.\dR\dR$) belong to $L$.
    
    After substituting the appropriate values, we~obtain the equation
    \begin{equation}
    \standardMeasure{L} = \frac{1}{2} + \frac{1}{8} \standardMeasure{L}^4
    \end{equation}
    which, by~the \emph{rational root theorem}, see e.g.~\cite{rationalRootTheoremBook} page~116, has no~rational solutions.
    
    To~conclude the proof, we~will describe how to~define the language~$L'$. The crux of~the construction comes 
    from the beautiful example by~Potthoff, see~\cite[Lemma 5.1.8]{potthoffExample}. 
    We~will use the following interpretation of~the lemma: one can define in~first\=/order logic
    over the signature $\{a,b, s_{\dL}, s_{\dR}, {\ancestor} \}$ that a~given finite tree\footnote{It is important to notice that the formula of Potthoff works over finite trees, i.e.~the fact that a~given tree is finite is an~assumption that is not expressed by the formula itself.} over the alphabet $\{a,b\}$ satisfies the following property:
    every node labelled~$a$ has exactly two children and every node labelled~$b$ is~a~leaf on~an~even depth.
    
    A discussion why the above language is in fact First\=/Order definable can be found in the proof of Theorem~13 on page~14 in~\cite{bojanczykTreeWalkingLATA}. The rough idea is that one can express in FO the notion of \emph{zig-zags}: a pair of nodes $u\preceq v$ of a~tree forms a \emph{zig-zag} if the path between them changes direction at each step, i.e.~the consecutive directions are $\dL,\dR,\dL,\dR,\ldots$ or $\dR,\dL,\dR,\dL,\ldots$. Based on that, one can express an inductive condition, that guarantees that all leafs are at the same depth modulo $2$. Finally, there is a unique leaf that is connected by a zig-zag with the root. Based on the shape of that zig-zag one can express the length of it modulo $2$.
    
    To construct~$L'$ we~simply utilise the formula defining the language in~the Potthoff's example to~define~$L'$
    by~substituting:
    \begin{enumerate}
	\item the formula describing a~leaf with the formula describing a~first occurrence of the label~$b$ on~a path: $\varphi_{\text{leaf}}(x) \eqdef b(x) \land \forall y.\ (y {\weakAncestor} x) {\implies} a(y)$,
    \item the formula describing an~internal node with the formula describing a~node labelled $a$ with no occurrences of the label~$b$ on~the path: $\varphi_{\text{node}}(x) \eqdef a(x) \land \forall y.\ (y {\weakAncestor} x) {\implies} a(y)$,
    \end{enumerate}
    
    Note that the set~$L'$ agrees with $L$~on every tree that 
    has a~label~$b$ on~every infinite path from the root, because such trees are interpreted as~finite trees.
    On~the other hand, the truth value of~the modified formula on~trees that have an~infinite path from the 
    root with no~nodes labelled~$b$, i.e.~on~the set $L_{a2}$ from Example~\ref{ex:someMeasuresOfSets}, is~of~no~concern to~us.
    Indeed, as~previously shown, the uniform measure of~the set~$L_{a2}$ is~$0$.
    
    To~be more precise, for every tree $\treeZ \in \treesInfinite{\{a,b\}} \setminus L_{a2}$ we~have that
    $\treeZ \in L \iff \treeZ \in L'$, where $L_{a2}$~is~the~language from Example~\ref{ex:someMeasuresOfSets}.
    Therefore, we~have that $L \cup L_{a2} = L' \cup L_{a2}$. 
    Since $\standardMeasure{L_{a2}} = 0 $, we~have that
    \[
    \standardMeasure{L} = \standardMeasure{L \cup L_{a2}} = \standardMeasure{L' \cup L_{a2}} = \standardMeasure{L'},
    \]
    which concludes the proof.
\end{proof}

\section{Conjunctives queries}
\label{sec:CQsMeasure}

Proposition~\ref{prop:algebraicValueFO} from the previous section implies that allowing the descendant relation in~full first\=/order logic
permits irrational values of~measures. Nevertheless, we~can allow use of the ancestor relation and retain both rational values and computability when we restrict the formulae to~the positive existential fragment using only atomic formulae and conjunction, i.e.~to the \emph{conjunctive queries}.

Recall that introducing the ancestor/descendant relation to~the tree structure causes that every two nodes in~the Gaifman graph are in~distance at~most two from each other.
Thus, for the purpose of having a~relevant definition of the distance in the tree, we retain the child related notion of distance,
i.e.~in~this section, as before, the notion of the distance is induced by the child relations only.

\paragraph{Conjunctive queries}
A~\emph{conjunctive query} (CQ) over an~alphabet $\Gamma$ is a~formula of first\=/order logic, using only conjunction and existential quantification, over unary 
predicates $a(x)$, for $a \in \Gamma$, the root predicate $\rootP(x)$, and binary predicates $s_{\dL}(x,y)$, $s_{\dR}(x,y)$, $s(x,y)$, and $\ancestor(x,y)$.

An~alternative way of looking at conjunctive queries is via graphs and graph homomorphisms.
Intuitively, a~conjunctive query can be seen as a~graph (a relational structure) in which the variables of the query constitute the vertices, the unary relations of the query label the vertices, and the binary relations form and label the edges. We call such graphs \emph{patterns}.
More formally, a~pattern $\pi$ over $\Gamma$ is 
a relational structure $\pi = \langle V, V_{\rootP}, E_{\dL}, E_{\dR}, E_s, E_{\ancestor}, \lambda_{\pi} \rangle$, where $\parfun{\lambda_{\pi}}{V}{\Gamma}$
is a partial labelling, $V_{\rootP}$ is the set of root vertices, and $G_\pi = \langle V, E_{\dL} \cup E_{\dR} \cup E_s \cup E_{\ancestor} \rangle$ is a finite graph whose edges are split into
left child edges $E_{\dL}$, right child edges $E_{\dR}$, child edges $E_{s}$, and ancestor edges $E_{\ancestor}$. By $|\pi|$ we mean the size of the underlying graph.

We say that a tree $t = \langle \dom{t}, s_{\dL}, s_{\dR}, \ancestor, (a^{t})_{a\in\Gamma} \rangle$ satisfies a pattern $\pi = \langle V, V_{\rootP}, E_{\dL}, E_{\dR}, E_s, E_{\ancestor}, \lambda_{\pi} \rangle$, denoted $t \models \pi$, if there exists a homomorphism $\fun{h}{\pi}{t}$, that is a function $\fun{h}{V}{\dom{t}}$ such that
\begin{enumerate}
    \item \fun{h} {\langle V, E_{\dL}, E_{\dR}, E_s, E_{\ancestor}\rangle} {\langle \dom{t}, s_{\dL}, s_{\dR}, s_{\dL} \cup s_{\dR} , \ancestor \rangle} is a homomorphism of relational structures,
    \item for every $v \in V_{\rootP}$ we have that $h(v) = \varepsilon$,
    \item and for every $v \in \dom{\lambda_{\pi}}$ we have that $\lambda_{\pi}(v) = t(h(v))$.
\end{enumerate}

Observe that, by definition, every conjunctive query can be represented as a pattern. The reverse is also true.
To obtain a~query introduce a~variable for every vertex of the pattern, and then express every vertex label by an~unary atom and every edge label by a~binary atom.
Since every pattern can be seen as a conjunctive query and vice versa, we will use those terms interchangeably.
The class of conjunctive queries is denoted \text{CQ}, the class of formulae that are Boolean combinations of conjunctive queries is denoted \text{BCCQ}.

\begin{thm}
    \label{thm:CQsAreRational}
    Let $q$ be~a~conjunctive query over the signature $\{\rootP, s_\dL, s_\dR, s,\allowbreak {\ancestor}\}\cup \Gamma$.
    Then, the uniform measure of~the language\footnote{A conjunctive query is a special first\=/order formula, thus the set $\lang(q)$ is well\=/defined.}~$\tL(q)$    
    is~rational and computable in~exponential space.
\end{thm}

To~prove the theorem we~will modify the concept of \emph{firm} sub\=/patterns, used e.g.~in~\cite{murlakAcyclic}.
Intuitively, a~\emph{firm sub\=/pattern} is~a~maximal part of~a~conjunctive query that has to~be~mapped in~a~small neighbourhood.
The overall proof strategy is~similar to~the first\=/order case: we identify those parts (in the form of \emph{firm} sub\=/patterns) of the conjunctive query 
that have to be satisfied in~the small neighbourhood of the root and those parts that can be satisfied arbitrarily far from the root.
As previously, the former decide the value of the measure and the latter can be ignored.


A~sub\=/pattern~$\pi'$ is~\emph{firm} if~it~is~a~sub\=/pattern of~a~pattern~$\pi$ induced by~vertices belonging to~a~maximal strongly\=/connected 
component in~the \emph{graph of connections}~$C_{\pi} = \langle V , E \rangle$ such that $V$ is the set of vertices of $G_\pi$ and $\langle x,y \rangle \in E$ if~either $\rootP(x)$, $x s_{\dL} y$, $y s_{\dL} x$, $x s_{\dR} y$, $y s_{\dR} x$, $x s y$, $y s x$, or $x {\ancestor} y$.
In~particular, a~pattern is~firm if~it~has a~single strongly\=/connected component.
We~say that a~sub\=/pattern is~\emph{rooted} if~it~contains the~predicate~$\rootP$.

\begin{proposition}
    \label{prop:firmPatternsShortDistance}
    Let $\pi = \langle V, V_{\rootP}, \allowbreak E_{\dL}, \allowbreak E_{\dR}, \allowbreak E_s, E_{\ancestor}, \lambda_{\pi} \rangle$ be~a~firm pattern. Then, for every tree~$t$ such that $t \models \pi$, for every two vertices~$x,y$ in~$V$, and for every homomorphism~$\fun{h}{\pi}{t}$
    we~have that $\td(h(x),h(y)) < |\pi|$.
    Moreover, if~$\pi$~is~rooted then for every vertex~$x$ we~have that~$\td(h(x),\varepsilon) < |\pi|$.
\end{proposition}
\begin{figure}
	\centering
	\begin{tikzpicture}[scale=1.9]
	
	\draw[dotted, blue] 
	(0,0) -- (-3,-3)
	(0,0) -- (3,-3);
	
	\node (v1) at (0,0) {};
	\node (v2) at (-.5,-.5) {};
    \node (v22)at (0,-1) {}; 
	\node (v3) at (-0.5,-1.5) {$h(y_{j+1})$};
	\node (v4) at (0.5,-1.5) {};
	\node (v5) at (1,-2) {};
	\node (v6) at (1.5,-2.5) {};
	\node (v7) at (1,-3) {};
	\node (v8) at (1.5,-3.5) {$h(x)$};
    \node (v9) at (-1,-2) {$u$};
    \node (v10) at (-1.5,-2.5) {$h(y_{j})$};
    \node (v11) at (-1.5,-3) {};
    \node (v12) at (-1,-3.5) {$h(y)$};
	
	\draw[dotted, gray] 	 
    (v1.center) -- (v2.center)
	(v2.center) -- (v22.center)
	(v22.center) -- (v3.center)
	(v4.center) -- (v5.center)
	(v5.center) -- (v6.center)
	(v6.center) -- (v7.center)
	(v7.center) -- (v8.center)
	(v8.center) -- (1,-4)
    
  	(v3.center) -- (v9.center)
    (v9.center) -- (v10.center)
    (v10.center) -- (v11.center)
    (v11.center) -- (v12.center)
    (v12.center) -- (-.5,-4);
    
    \draw[dashed, black]
    (v22.center) -- (v3.center)
    (v22.center) -- (v4.center)
    (v4.center) -- (v5.center)
    (v5.center) -- (v6.center)
    (v6.center) -- (v7.center)
    (v7.center) -- (v8)
        
    (v3) -- (v9)
    (v9) -- (v10)
    (v10) -- (v11)
    (v11) -- (v12);

	\end{tikzpicture}
    \caption{A~possible placement of~nodes in~the proof of~Proposition~\ref{prop:firmPatternsShortDistance}.}
    \label{fig:firm-pattern-short-distance}
\end{figure}
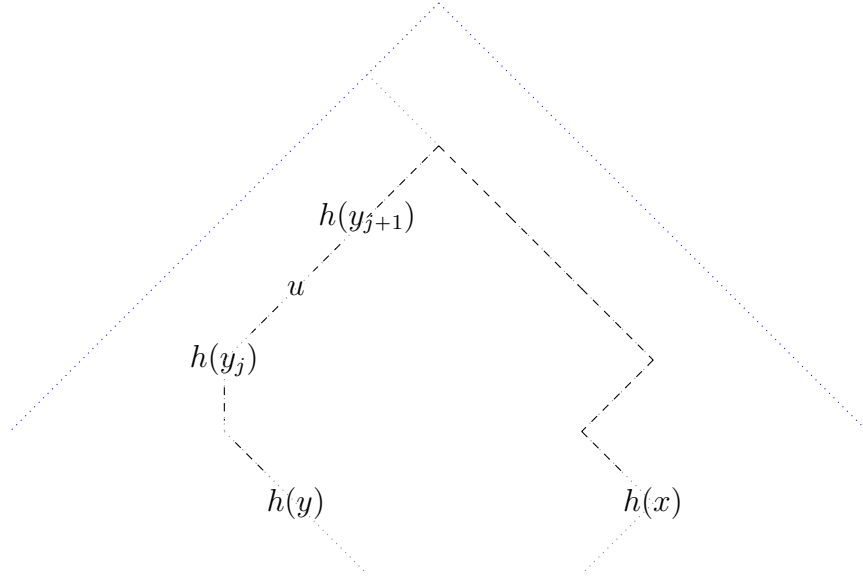

\begin{proof}
    Let us~assume otherwise and put $n = |\pi|$. Then, there is~a~tree $t$, a~homomorphism $h$, and two vertices $x$, $y$ such that $t \models \pi$ and $\td(h(x),h(y)) \ge n$.
    We claim that $x$ and $y$ cannot be~in~the same strongly\=/connected component.
    
    Since for some~$m$ we~have that $\td(h(x),h(y)) = m - 1 \ge n$, there is~a~simple path connecting $h(x)$ and $h(y)$. That is, there is~a~sequence of~distinct nodes~$u_1, u_2, \dots, u_m$ such that $u_1 = h(x)$,  $u_m = h(y)$ and for every~$i$,
    $u_i$ and $u_{i+1}$ are in a~child relation, i.e.~$s(u_{i}, u_{i+1})$ or~$s(u_{i+1}, u_{i})$.
    Notice that every node in the sequence is an~ancestor of one of the nodes $h(x)$ or $h(y)$.
    Moreover, since the path consists of $n+1$ nodes, there has to~be~a~node $u$ in~the sequence such that $u$ is~not in~the image of~the homomorphism~$h$.
    More precisely, there is~a~node~$u$ such that $u = u_i$ for some $1 \leq i \leq m$, $u \notin h(\pi)$, and
    one of~the nodes $h(x)$ or~$h(y)$ is~a~descendant of~$u$, i.e.~$u \ancestor h(x)$ or~$u \ancestor h(y)$.
    Without loss of generality, let us~say that~$u \ancestor h(y)$. Or, more precisely, that $u\dL \weakAncestor h(y)$.
    See Figure~\ref{fig:firm-pattern-short-distance} for a~possible placement of the nodes.
    
    If~$x$ and $y$ were in~the same strongly\=/connected component then there would be~a~path that connects $y$~to~$x$ in~the graph of~connections~$C_\pi$,
    i.e.~a~sequence of~vertices $y_1, y_2, \dots, y_k$, for some $k$, such that $y_1 = y$, $y_k = x$, and for every $i= 1, \dots, k-1$ there is~an edge between $y_i$ and $y_{i+1}$ in $C_{\pi}$.
    In particular, this would imply that for every~$i$ we~have that $h(y_i)$ and $h(y_{i+1})$ are $\weakAncestor$\=/comparable.
    Now, there would also exist an~index $j \in \{1,\dots, k-1 \}$ such that $h(y_{j+1}) \ancestor u \ancestor h(y_{j})$.
    Indeed, if~there would be~no~such index, then all the vertices $y_i$ would satisfy $u\dL \weakAncestor h({y_i})$, as $y_i$ and $y_{i+1}$ are $\weakAncestor$\=/comparable
    for every index~$i$.
    But this is~impossible because if~$u\dL \weakAncestor h(y_i)$ for all~$i$, then we~would have that $u\dL \weakAncestor h(y_k) = h(x)$.
    Now, since $u\dL \weakAncestor h(y)$ and $u\dL \weakAncestor h(x)$, then, by~the definition of~the distance, $u$~would not belong to~the sequence $u_1, \dots, u_m$.
    Which is~a~contradiction with our assumption.
    
    Therefore, there is~an~index~$j$ such that $h(y_{j+1}) \ancestor u \ancestor h(y_{j})$.
    Thus, by~the definition of~$C_{\pi}$ we~have that either $\rootP(y_{j})$,
    $y_{j} s_{\dL} y_{j+1}$, $y_{j+1} s_{\dL} y_{j}$, $y_{j} s_{\dR} y_{j+1}$, $y_{j+1} s_{\dR} y_{j}$, $y_{j} s y_{j+1}$, $y_{j+1} s y_{j}$, or $y_{j} {\ancestor} y_{j+1}$.
    Neither of~child relations is~possible because the distance between 
    $h(y_{j})$ and $h(y_{j+1})$ is~at~least two. Similarly, both $y_{j} {\ancestor} y_{j+1}$ and $\rootP(y_{j})$ are impossible because
    we~have that $u \ancestor h(y_{j})$.
    Hence, there can be~no~such sequence $y_1, \dots, y_k$ and we obtain a~contradiction. Thus $x$ and $y$ cannot belong to~the same strongly\=/connected component.
    This proves the first part of~the lemma.
    
    Now, if~$\pi$ is~rooted then there is~a~vertex~$y$ such that for every homomorphism~$h$ we~have that $h(y) = \varepsilon$.
    Hence, by~the first part of~the lemma, for all vertices $x\in \pi$ we~have that $\td(h(x),\varepsilon) = \td(h(x),h(y))< n$.
\end{proof}

\paragraph{Graph of firm sub\=/patterns}
Let $q$ be~a~conjunctive query.
Consider a~graph $\graph^F_q = \langle V, E \rangle $ where $V$ is~the set of~firm sub\=/patterns of~$q$ and there is~an~edge
$\langle v_1, v_2 \rangle \in E \subseteq V \times V$ 
between two vertices $v_1$, $v_2$ if~and only if~there is~an~$\ancestor$\=/labelled edge between some two vertices $w_1 \in v_1$, $w_2 \in v_2$.
We~call this graph the \emph{graph of~firm sub\=/patterns} of~the conjunctive query~$q$.

\begin{fact}
    The directed graph $G^F_q$ of~firm sub\=/patterns of~a~conjunctive query~$q$ is~acyclic and has at~most one rooted firm sub\=/pattern. We~call that sub\=/pattern the~\emph{root sub\=/pattern}.
\end{fact}

\begin{proof}
    By the definition of~the firm sub\=/patterns, every node with the predicate~$\rootP$ ends up~in~the same maximal strongly\=/connected component.
    The acyclicity follows directly from the fact that firm sub\=/patterns are the maximal strongly\=/connected components.
\end{proof}

As~in~the case of~root formulae, the root pattern decides of~the behaviour of~a~satisfiable conjunctive query,
as~expressed by~the following lemma.

\begin{lemma}
    \label{lemma:rootPatternExtraction}
    Let $q$ be~a~conjunctive query over the signature $\{\rootP, s_\dL, s_\dR, s, {\ancestor} \}\cup \Gamma$.
    Then, either
    \begin{itemize}
        \item $q$ is~not satisfiable and~$\standardMeasure{\tL(q)} = 0,$
        \item $q$ is~satisfiable,  has no~root sub\=/pattern, and $\standardMeasure{\tL(q)} = 1,$
        \item or $q$~is~satisfiable,  has a~root sub\=/pattern $p$, and $\standardMeasure{\tL(q)} = \standardMeasure{\tL(p)}$.
    \end{itemize}
\end{lemma}

\begin{proof}
    If $q$ is~not satisfiable then $\tL(q) = \emptyset$ and so~$\standardMeasure{\tL(q)} = 0$.
    Let $q$ be~satisfiable, i.e.~there is~a~tree~$\treeZ^q$ and a~homomorphism~$\fun{h}{\graph_q}{\treeZ^q}$.
    Let $\treeZ^r$ be~a~finite tree such that $h(\graph_q) \subseteq \nodes{\treeZ^r}$ and let the set $S \subseteq \treesInfinite{{\Gamma}}$ be~the set of~all trees $\treeZ$ such that for every node $u \in \dS^{|q| +1}$ the tree $\treeZ^r$ is~a~prefix of~$\subtreeAtNode{\treeZ}{u'}$ for some node $u'$ such that $u\weakAncestor u'$. By~Lemma~\ref{lemma:basicMeasuresLemma}, we~have that~$\standardMeasure{S} = 1$.
    
    If $q$~has no~root firm sub\=/pattern, then $S \subseteq \tL(q)$ and we~have that 
    \[\standardMeasure{\tL(q)} \geq \standardMeasure{S} = 1.\]
    
    On~the other hand, if $q$ has a~root sub\=/pattern~$p$ then for every tree~$\treeZ \in S$
    we~have that $\treeZ \models q$ if~and only if~$\treeZ \models p$.
    Thus, $\tL(q) \cap S = \tL(p) \cap S$ and since $\standardMeasure{S} = 1$, we~have that 
    \[
    \standardMeasureBig{\lang(q)} = \standardMeasure{\tL(q) \cap S} = \standardMeasure{\tL(p) \cap S} = \standardMeasureBig{\lang(p)}.\qedhere
    \]
\end{proof}

In~other words, the problem of~computing the uniform measure of~a~set of~full trees defined by~a~conjunctive query reduces to~the following problem of~counting the models of~a~fixed height.

\decpro{Conjunctive queries counting}{problem:cqCounting}{A~conjunctive query $q$ and a natural number $n$ in unary.}{Number of full trees of height $n$ that satisfy $q$.}

\begin{proposition}
    Problem~\ref{problem:cqCounting} can be solved in~space exponential in~$n$ and 
    polynomial in~the size of~the query~$q$.
\end{proposition}

\begin{proof}
    All we~need is~to enumerate all full trees of~height~$n$ and check whether they satisfy the query.
    The number of~such trees is~exponential in~$n$ and the model checking can be~done in~polynomial space, with respect to~both the query and the size of~the tree, see Lemma~\ref{lemma:modelCheckingFO}. 
\end{proof}

We~infer Theorem~\ref{thm:CQsAreRational} as~an~immediate consequence.

\begin{proof}[Proof of~Theorem~\ref{thm:CQsAreRational}]
    In~polynomial space we~can check whether the query is satisfiable, see e.g.~\cite{bjorklundCQscontainment},
    and in~polynomial time compute its root sub\=/pattern: it~is~folklore that one can compute all strongly\=/connected components of~a~directed graph in~polynomial time.
    
    Now, by~Lemma~\ref{lemma:rootPatternExtraction}, if~the query is~not satisfiable then the measure is~$0$;
    if~it~is~satisfiable, but there is~no~root sub\=/pattern then the measure is~$1$.
    Thus, the only case left is~when the query is~satisfiable and has a~root sub\=/pattern~$p$.
    
    If~it~is~the case, then $\standardMeasure{\tL(q)} = \standardMeasure{\tL(p)}$.
    Since $p$ is~a~root sub\=/pattern, then by~Proposition~\ref{prop:firmPatternsShortDistance} for any tree~$t$, any homomorphism $\fun{h}{p}{t}$, and any vertex $v$ of~the query we have that $|h(v)| < |p|$.
    Thus, for any full tree $t \in \treesInfinite{\Gamma}$ to~decide 
    whether $t \models p$ we~only need to~check the prefix of~$t$ that is~of~height $|p|$.

    Since the sets $\mathbb{B}_{t'}$, where $t'$ ranges over full trees of~height $|p|$, form a~partition of~$\treesInfinite{\Gamma}$, to~compute $\standardMeasure{\tL(p)}$ it~is~enough to~iterate
    over all such trees and compute how many of~them satisfy~$p$.
    
    Since every full tree of~height~$|p|$ is~of~exponential size with respect to~$p$ and 
    they can be~iterated in~exponential space, we~infer that the measure~$\standardMeasure{\tL(p)}$ can be~computed 
    in~exponential space.
    To conclude the proof we notice that $\standardMeasure{\tL(p)}$ is a~finite sum of~the measures of~sets of form~$\setBall{t}$, where  $t$ is~some finite tree. Since for a~finite tree $t$ the measure of~$\setBall{t}$ is~rational, then so is $\standardMeasure{\tL(p)}$.
\end{proof}

As~in the case of~first\=/order formulae, we~can lift Theorem~\ref{thm:CQsAreRational} to~Boolean combinations of~conjunctive queries.

\begin{corollary}
    \label{cor:booleanCombinationOfCQsAreRational}
    Let $\varphi$ be a Boolean combination of conjunctive queries. Then, the uniform measure $\standardMeasureS(\tL(\varphi))$ is~rational and can be~computed in~exponential space.
\end{corollary}

\begin{proof}
    This is~an~immediate consequence of~Lemma~\ref{lemma:booleanAlgebraAndMeasure} and Theorem~\ref{thm:CQsAreRational}.
    
    Indeed, by~Lemma~\ref{lemma:booleanAlgebraAndMeasure} and Theorem~\ref{thm:CQsAreRational}, patterns in~$\varphi$ can be~replaced by~their root patterns without the change of~measure, or~by the query~$\bot$ if they are unsatisfiable.    
    The rest of~the reasoning follows as~in~the proof of~Theorem~\ref{thm:CQsAreRational}.
\end{proof}

We can slightly strengthen the above result if we want to only decide whether the measure is positive.

\begin{proposition}
	\label{prop:upperPositiveBCCQ}
	The positive $\standardMeasure{\text{BCCQ}}$ problem is  in \nexp.
\end{proposition}

\begin{proof}

	Let $\phi$ be a Boolean combination of conjunctive queries. Let $m$ be the maximum over the sizes of the conjunctive queries in $\phi$.
	By Lemma~\ref{lemma:rootPatternExtraction} and Lemma~\ref{lemma:booleanAlgebraAndMeasure}, we can translate $\phi$ into a Boolean combination of firm, rooted conjunctive queries $\phi^{\ast}$ such that $\standardMeasure{\lang(\phi)} = \standardMeasure{\lang(\phi^{\ast})}$ and $\phi^{\ast}$ is of~polynomial size with respect to~$\phi$.
	This can be done in exponential time and requires verifying whether the patterns in $\phi$ are satisfiable.
	
	Now, since every conjunctive query in $\phi^{\ast}$ is firm and rooted, then either there is a~finite tree~$t$ of depth $2n$ such that $t \models \phi^\ast$ or~$\phi^{\ast}$ is~not satisfiable. If~there is~no such tree, then $\standardMeasure{\lang(\phi^{\ast})} = 0$. If~such a~tree $t$ exists, then $\setBall{t} \subseteq \lang(\phi^{\ast})$ and by~Lemma~\ref{lemma:basicMeasuresLemma} we~have that $\standardMeasure{\lang(\phi^{\ast})} > 0$.
	Hence, it is enough to guess the tree $t$ and verify that $t \models \phi^\ast$. Since $t$ is of exponential size in $\phi$ and model checking 
	of a conjunctive query can be done in polynomial time, we infer the upper bound.
\end{proof}

\section{Non\=/deterministic safety automata}
\label{sec:open-sets}

In this section we provide an~effective method of~computing the uniform measure of those regular sets of~infinite trees which are recognisable by non\=/deterministic safety automata. These automata, also known as automata \emph{with no acceptance condition} can be equivalently defined as parity automata that use only the parity $0$, i.e.~have Rabin--Mostowski index equal to $(0,0)$. The result of this section is expressed by the following theorem.

\begin{thm}
    \label{thm:openSetsUpperBound}
    The measure $m=\standardMeasureBig{\lang(\AA)}$ is an~algebraic number that can be effectively computed given a~safety automaton $\AA$. Moreover, this computation can be done in three\=/fold exponential time and the decision problems about $m$ (e.g. $m{>}0$, $m{>}0.5$, etc) can be solved in two\=/fold exponential space.
\end{thm}

Before providing a formal definition of the considered class of languages, we recall a~standard fact, that is often considered folklore, see~\cite[Section~6]{walukiewicz_low_levels}, \cite[Section~3.2]{cavallari_phd}, or~\cite{janin_closed}.

\begin{proposition}
\label{pro:closed-is-safety}
The following conditions are equivalent for a regular set of infinite trees $L$.
\begin{enumerate}
\item $L$ can be recognised by a non\=/deterministic safety automaton,
\item $L$ is closed as a subset of $\treesInfinite{\Gamma}$.
\end{enumerate}
\end{proposition}

Thus, it can be said that the procedure from Theorem~\ref{thm:openSetsUpperBound} works for languages that are topologically simple.

The idea behind the procedure is as follows. If $L$ is a~regular closed set, then $L$ is an~intersection $\bigcap_{n\in\bbN} L_n$ of a~decreasing sequence of regular sets $L_n$ that are Boolean combinations of base sets. Moreover, that sequence can be effectively computed.
Intuitively, the language $L_n$ in the above intersection describes the trees in $L$ \emph{up to the depth $n$}.

Having found an~appropriate decomposition as above, we know that $\standardMeasureBig{L}=\lim\limits_{n\to\infty} \standardMeasureBig{L_n}$, i.e.~the sequence of sets $L_n$ approximates $L$ with respect to the measure.
It turns out that to compute the measure of $L$ it is enough to be able to track changes of the measures of the subsequent approximations. To do so, we introduce a~finitely dimensional space $\DD\subseteq\bR^k$ (for appropriately chosen $k$) and show that there exist: an~initial value $\alpha_0\in\DD$ and two computable functions $\fun{\FF}{\DD}{\DD}$ and $\fun{\MM}{\DD}{\bR}$, such that:
\begin{equation}
\standardMeasureS \big(L_n\big)=\MM\big(\FF^n(\alpha_0)\big).
\end{equation}
Moreover, $\FF$, $\MM$, and $\alpha_0$ can be chosen so that $\FF$ is continuous, monotone (w.r.t.~an appropriately chosen order on $\DD$), and defined by a~vector of polynomials; $\MM$ is continuous; and $\alpha_0$ is the greatest element of $\DD$. Therefore, repeating a~variant of Kleene's Fixpoint Theorem, we infer that $\lim_{n\to\infty} \standardMeasureBig{L_n}$ is equal to $\MM(\alpha_\infty)$, where $\alpha_\infty\in\DD$ is the greatest fixpoint of $\FF$. Now, since the function $\FF$ is given by a~vector of polynomials, the value of $\alpha_\infty$ can be effectively computed using Tarski's Quantifier Elimination.

\paragraph*{Notation}
For the sake of this section, we will use the following notation, for $d\in\bbN$ and $t\in\treesInfinite{\Gamma}$:
\[t^{\leq d}\eqdef t\!\!\upharpoonright\!\!_{\{\dL,\dR\}^{\leq d}}\in\trees{\Gamma}^{d},\]
i.e.~the tree $t^{\leq d}$ is the full tree of height $d$ obtained by restricting $t$ to the nodes of height at most $d$.

Additionally, we will consider the set $\trees{\Gamma}^{d}$ for $d\in\bbN$ as a~measurable space with the uniform discrete measure $\mu^{d}$, where the probability of each finite tree $\tau\in\trees{\Gamma}^{d}$ equals $|\Gamma|^{-|\dom{\tau}|}$. Notice that for each subset $S\subseteq \trees{\Gamma}^{d}$ we have
\begin{equation}
\mu^{d}(S)=\standardMeasureBig{\{t\in\treesInfinite{\Gamma}\mid t^{\leq d} \in S\}}=\standardMeasureBig{\bigcup_{\tau\in S} \setBall{\tau}}.
\label{eq:partial-measure}
\end{equation}

\subsection{Safety automata}
\label{sec:safety-auto}

We begin the proof by providing a formal definition of safety automata and their properties. A~\emph{safety automaton} is a~tuple $\AA=\langle \Gamma, Q, \delta, I\rangle$ where: $\Gamma$ is an~alphabet, $Q$ is a~finite set of \emph{states}, $\delta\subseteq Q\times \Gamma\times Q\times Q$ is a~\emph{transition relation}, and $I\subseteq Q$ is a~set of \emph{initial states}. Let $t$ be a~tree. A~\emph{run} of an~automaton $\AA$ over $t$ is a function $\fun{\rho}{\dom{t}}{Q}$ such that for each $v\in\dom{t}$ if both $v\dL$ and $v\dR$ belong to $\dom{t}$ then
\[\big(\rho(v), t(v), \rho(v\dL),\rho(v\dR)\big)\in\delta.\]
Notice that if $t$ is a~finite tree of height $0$, i.e.~$\dom{t}=\{\varepsilon\}$ then every function $\fun{\rho}{\{\epsilon\}}{Q}$ is a~run of $\AA$ over $t$. A~run $\rho$ is \emph{accepting} if $\rho(\epsilon)\in I$. We say that $\AA$ \emph{accepts} a~tree $t$ if there exists an~accepting run of $\AA$ over $t$. For $d\in\bbN$ by $\lang^{d}(\AA)\subseteq\trees{\Gamma}^{d}$ we denote the \emph{language} of the automaton $\AA$, that is the set of full trees of height $d$ that are accepted by $\AA$. Similarly, $\lang(\AA)\subseteq\treesInfinite{\Gamma}$ is the set of full trees accepted by $\AA$.

Fix a~safety automaton $\AA$ and let $t$ be a~tree over $\Gamma$. The \emph{type} of $t$, denoted $\type(t)\subseteq Q$ is defined as the set of states $q\in Q$ such that there exists a~run $\rho$ of $\AA$ over $t$ with $\rho(\varepsilon)=q$. Thus, $\AA$ accepts $t$ if and only if $\type(t)\cap I\neq\emptyset$.

The following lemma is a~standard application of a~compactness argument.

\begin{lemma}
    \label{lem:safety-representation}
    Let $t\in\treesInfinite{\Gamma}$ be a~full tree. Then the following conditions are equivalent
    \begin{enumerate}
        \item $t\in\lang(\AA)$,
        \item for every $d\in\bbN$ we have $t^{\leq d}\in\lang^d(\AA)$.
    \end{enumerate}
\end{lemma}

\begin{proof}
The implication $1.\Rightarrow 2.$ is clear from the definition. Consider the opposite implication. Let $t$ be a~full tree such that for all $d\in\bbN$ we have $t^{\leq d}\in\lang^d(\AA)$. Assume that $\rho_d\in \trees{Q}^{d}$ is a~run witnessing that. Then, there exists a sub\=/sequence of the sequence of runs $(\rho_d)_{d\in\bbN}$ that is point\=/wise convergent in $\trees{Q}$. Let $\rho_\infty$ be the limit of such a~sub\=/sequence. Then $\rho_\infty\in\trees{Q}^{\w}$ and it is a run of $\AA$ over $t$. Moreover, $\rho_\infty(\varepsilon)\in I$ because it is the case for each of the runs $\rho_d$. Therefore, $t\in\lang(\AA)$.
\end{proof}

We will now observe that each safety automaton $\AA$ can be lifted to another safety automaton (denoted $\widehat{\AA}$) that recognises the same language but its transition relation is bottom\=/up functional. More precisely, consider a~safety automaton $\AA$ and let $\widehat{\AA}\eqdef\left\langle \Gamma, \powerset(Q), \widehat{\delta},\widehat{I}\right\rangle$, where $\fun{\widehat{\delta}}{\Gamma\times \powerset(Q)\times \powerset(Q)}{\powerset(Q)}$ is defined as
\[\widehat{\delta}(a,R_\dL,R_\dR)\eqdef \big\{q\in Q\mid \exists (q,a,q_\dL,q_\dR)\in\delta.\ q_\dL\in R_\dL\land q_\dR\in R_\dR\big\},\]
and $\widehat{I}$ contains those $R\subseteq Q$ such that $R\cap I\neq\emptyset$. For the sake of simplicity, we will use $\fun{\widehat{\delta}_a}{\powerset(Q)^2}{\powerset(Q)}$ to denote the function $\widehat{\delta}$ with the first argument fixed to $a$.

\begin{remark}
Directly from the definition, we have $\lang(\AA)=\lang(\widehat{\AA})$.
\end{remark}

In the following constructions we will treat the set $\powerset(Q)$ as partially ordered by the inclusion ${\subseteq}$. The following fact follows directly from the definition of $\widehat{\delta}$.

\begin{fact}
    \label{ft:delta-mono}
    The function $\widehat{\delta}_a$ is monotone on both its arguments.
\end{fact}

\subsection{Fixpoints}

In this section we will develop tools allowing to define the previously mentioned space~$\DD$ and the functionals $\FF$ and $\MM$. Let us fix a~safety automaton $\AA$. Let $\DD$ be the space of probability distributions over $\powerset(Q)$. For each $R\subseteq Q$ let $\one_R\in\DD$ be the distribution concentrated in $R$, i.e.~$\one_R(R)=1$ and for each $R'\neq R$ we have $\one_R(R')=0$.

We will now lift the inclusion order ${\subseteq}$ on $\powerset(Q)$ to an~order over $\DD$. Recall that a~set $\Uu\subseteq\powerset(Q)$ is \emph{upward\=/closed} if whenever $R\in\Uu$ and $R\subseteq R'$ then $R'\in\Uu$ as~well. Consider two probability distributions $\alpha,\beta\in\DD$ over $\powerset(Q)$. Let $\alpha\preceq \beta$ hold if and only if, for every upward\=/closed set $\Uu\subseteq \powerset(Q)$, we have the following inequality
\begin{equation}
\label{eq:dist-order}
\sum_{R\in \Uu} \alpha(R)\leq \sum_{R\in \Uu} \beta(R).
\end{equation}
Intuitively, the partial order ${\preceq}$ over $\DD$ corresponds to the process of distributing a given distribution $\beta$ over smaller sets. Such orders are known as~\emph{probabilistic powerdomains}, see e.g.~\cite{probabilisticPowerdomains}.

Notice that $\one_R\preceq \one_S$ if and only if $R\subseteq S$. Also, it is easy to observe that $\one_Q$ is the greatest element of $\DD$.

Our aim is to simulate the process of randomly choosing letters $a\in \Gamma$ (as done in the probabilistic spaces $\trees{\Gamma}^{d}$) as a~monotone mapping $\FF$ over the distributions $\DD$. Fix a~letter $a\in \Gamma$. Recall that $\fun{\widehat{\delta}_a}{\powerset(Q)^2}{\powerset(Q)}$ is the transition function of the power\=/set automaton $\widehat{\AA}$. Define $\fun{\FF_a}{\DD}{\DD}$ as
\begin{equation}
\FF_a(\alpha)(R)=\sum_{(R_\dL,R_\dR)\in\widehat{\delta}_a^{-1}(R)}\ \alpha(R_\dL)\cdot \alpha(R_\dR).
\label{eq:def-of-F}
\end{equation}

In other words, $\FF_a(\alpha)$ is the distribution over $\powerset(Q)$ given by the values of~$\widehat{\delta}_a$ where its arguments are chosen independently according to the distribution $\alpha$. Finally, let $\fun{\FF}{\DD}{\DD}$ be defined as $\FF(\alpha)=|\Gamma|^{-1}\cdot \sum_{a\in \Gamma} \FF_a(\alpha)$ (where the sum is taken coordinate\=/wise). It is easy to verify that all the functions $\FF_a$ and $\FF$ in fact produce probabilistic distributions, i.e.~elements of $\DD$.

\begin{lemma}
    The function $\FF$ is monotone with respect to the order ${\preceq}$ on $\DD$.
\end{lemma}

\begin{proof}
    Consider two distributions $\alpha\preceq \beta\in \DD$. Let $\Uu\subseteq \powerset(Q)$ be an~upward\=/closed set. Our aim is to prove the inequality~\eqref{eq:dist-order} for $\FF(\alpha)$ and $\FF(\beta)$. We will do it for each letter separately, what means that it is enough to show that for each $a \in \Gamma$ we have
    \[\sum_{R\in \Uu}\ \sum_{(R_\dL,R_\dR)\in\widehat{\delta}_a^{-1}(R)} \alpha(R_\dL)\cdot \alpha(R_\dR)\ \leq\ \sum_{R\in \Uu}\ \sum_{(R_\dL,R_\dR)\in\widehat{\delta}_a^{-1}(R)} \beta(R_\dL)\cdot \beta(R_\dR).\]
    Since $\widehat{\delta}_a$ is a function, the sums on both sides of the equality take the form
    \[\sum_{(R_\dL,R_\dR)\colon\widehat{\delta}_a(R_\dL,R_\dR)\in\Uu} \alpha(R_\dL)\cdot \alpha(R_\dR)\ \leq\ \sum_{(R_\dL,R_\dR)\colon\widehat{\delta}_a(R_\dL,R_\dR)\in\Uu} \beta(R_\dL)\cdot \beta(R_\dR).\]
    Which is equivalent to
\begin{equation}
\sum_{R_\dL\subseteq Q}\ \alpha(R_\dL)\cdot \sum_{R_\dR\colon\widehat{\delta}_a(R_\dL,R_\dR)\in\Uu} \alpha(R_\dR)\ \leq\ \sum_{R_\dR\subseteq Q}\ \beta(R_\dR)\cdot \sum_{R_\dL\colon\widehat{\delta}_a(R_\dL,R_\dR)\in\Uu} \beta(R_\dL).
\label{eq:enough-for-mono}
\end{equation}   
To obtain the last inequality, we observe that:
\begin{align*}
\sum_{R_\dL\subseteq Q}\ \alpha(R_\dL)\cdot \sum_{R_\dR\colon\widehat{\delta}_a(R_\dL,R_\dR)\in\Uu} \alpha(R_\dR)\ &\stackrel{(1)}{\leq}\
\sum_{R_\dL\subseteq Q}\ \alpha(R_\dL)\cdot \sum_{R_\dR\colon\widehat{\delta}_a(R_\dL,R_\dR)\in\Uu} \beta(R_\dR)\\
&\stackrel{(2)}{=} \sum_{R_\dR\subseteq Q}\ \beta(R_\dR)\cdot \sum_{R_\dL\colon\widehat{\delta}_a(R_\dL,R_\dR)\in\Uu} \alpha(R_\dL) \\
&\stackrel{(3)}{\leq} \sum_{R_\dR\subseteq Q}\ \beta(R_\dR)\cdot \sum_{R_\dL\colon\widehat{\delta}_a(R_\dL,R_\dR)\in\Uu} \beta(R_\dL),
\end{align*}
where the consecutive (in)equalities follow from:
\begin{enumerate}
\item[(1)] Fact~\ref{ft:delta-mono}: the family of sets $R_\dR$ is upward\=/closed and $\alpha\preceq \beta$;
\item[(2)] rearranging the sum;
\item[(3)] the fact that the family of sets $R_\dL$ is upward\=/closed and $\alpha \preceq \beta$ again.
\end{enumerate}
Thus, we have proved~\eqref{eq:enough-for-mono}. This concludes the proof of the lemma.
\end{proof}

Now we will see that, in a~sense, $\FF$ simulates the behaviour of $\AA$ over $\trees{\Gamma}^{d}$. Let $\alpha_0=\one_Q$ and $\alpha_{d+1}=\FF(\alpha_d)$ for $d \in \bbN$.

\begin{proposition}
    \label{pro:alpha-vs-types}
    For each $d\in\bbN$ the distribution $\alpha_d$ is the distribution of $\type(\tau)$ for a~random partial tree $\tau\in\trees{\Gamma}^{d}$, i.e.~for each $R\subseteq Q$ we have
    \[\alpha_d(R)=\mu^{d}\big(\{\tau\in\trees{\Gamma}^d\mid \type(\tau)=R\}\big).\]
\end{proposition}

\begin{proof}
    The proof is by induction. For $d=0$ the claim is obvious, as $\alpha_0=\one_Q$, and for each partial tree $\tau\in\trees{\Gamma}^0$ we have $\type(\tau)=Q$. Let us assume that the inductive hypothesis holds for $d$. Observe that
    \[\mu^{d+1}\big(\{\tau{\in}\trees{\Gamma}^{d+1}\mid \type(\tau){=}R\}\big)=\sum_{a\in \Gamma} \mu^{d+1}\big(\{\tau{\in}\trees{\Gamma}^{d+1}\mid \tau(\varepsilon){=}a\land\type(\tau){=}R\}\big).\]
    Now, as a~partial tree $\tau\in\trees{\Gamma}^{d+1}$ can be seen as a pair of partial trees $\tau_\dL,\tau_\dR\in\trees{\Gamma}^{d}$ merged by the letter $\tau(\varepsilon)$, we obtain that:
    \begin{align}
    &\mu^{d+1}\big(\{\tau\in\trees{\Gamma}^{d+1}\mid \tau(\varepsilon)=a\land\type(\tau)=R\}\big)=\nonumber\\ &\mu^{d}\times\mu^{d}\big(\{(\tau_\dL,\tau_\dR)\in\trees{\Gamma}^{d}\times\trees{\Gamma}^{d}\mid \widehat{\delta}_a\big(\type(\tau_\dL),\type(\tau_\dR)\big)=R\}\big),
    \label{eq:prob-sum}
    \end{align}
    where the latter probability is taken in the product space. Thus,~\eqref{eq:prob-sum} equals
    \[\sum_{(R_\dL,R_\dR)\in\widehat{\delta}_a^{-1}(R)} \mu^{d}\big(\{\tau_\dL\in\trees{\Gamma}^{d}\mid\type(\tau_\dL)=R_\dL\}\big)\cdot \mu^{d}\big(\{\tau_\dR\in\trees{\Gamma}^{d}\mid\type(\tau_\dR)=R_\dR\}\big),\]
    which, by the inductive assumption, equals
    \[\sum_{(R_\dL,R_\dR)\in\widehat{\delta}_a^{-1}(R)} \alpha_d(R_\dL)\cdot\alpha_d(R_\dR).\]
    The last term is equal to $\FF_a(\alpha_n)(R)$ by the definition, see~\eqref{eq:def-of-F}. Therefore, the inductive hypothesis holds for $d+1$.
\end{proof}

The next three statements follow the standard way of proving Kleene's Fixpoint Theorem. However, instead of simply invoking that result, we prove these claims by hand, as it seems to be much more direct.

\begin{fact}
    The sequence $(\alpha_d)_{d\in\bbN}$ is descending in the order ${\preceq}$ on $\DD$.
\end{fact}

\begin{proof}
    By monotonicity of $\FF$ and the fact that $\alpha_0$ is the greatest element of~$\DD$.
\end{proof}

\begin{lemma}
    There exists a~probabilistic distribution $\alpha_\infty\in\DD$ such that for each $R\subseteq Q$ we have
    \[\alpha_\infty(R)=\lim_{d\to\infty} \alpha_d(R).\]
\end{lemma}

\begin{proof}
    First notice that for each upward\=/closed set $\Uu\subseteq\powerset(Q)$ the sequence of $\sum_{R\in\Uu}\alpha_d(R)$ is non\=/increasing and bounded for $d\to\infty$. Thus, that sequence has a~limit. Now, for each $R\subseteq Q$ the value $\alpha_d(R)$ can be written as a~difference of sums as above for two upward\=/closed sets. Therefore, $\lim_{d\to\infty}\alpha_d(R)$ exists and we can define $\alpha_\infty(R)$ as that limit. Since $\sum_{R\subseteq Q}\alpha_d(R)$ is always~$1$, the limit values also satisfy that property and thus $\alpha_\infty$ is a probabilistic distribution.
\end{proof}

\begin{lemma}
    \label{lem:greatest}
    The distribution $\alpha_\infty$ is the greatest fixpoint of $\FF$, i.e.~$\FF(\alpha_\infty)=\alpha_\infty$ and if $\FF(\alpha')=\alpha'$ for some $\alpha'\in\DD$ then $\alpha'\preceq \alpha_\infty$.
\end{lemma}

\begin{proof}
    First notice that $\FF$ is continuous as a~function from $\bR^{\powerset(Q)}$ to $\bR^{\powerset(Q)}$. Therefore, \[\FF(\alpha_\infty)=\FF\left(\lim_{d\to\infty}\alpha_d\right)=\lim_{d\to\infty}\FF(\alpha_d)=\lim_{d\to\infty}\alpha_{d+1}=\alpha_\infty.\]
    Now consider any fixpoint $\alpha'\in\DD$ of $\FF$. By monotonicity of $\FF$ and the fact that $\alpha_0=\one_Q$, we know that for each $d\in\bbN$ we have $\alpha'\preceq \alpha_d$. As the order ${\preceq}$ is defined by sums of values and $\alpha_\infty$ is a point\=/wise limit of $\alpha_d$, it also holds that $\alpha'\preceq\alpha_\infty$.
\end{proof}

Recall that $\widehat{I}=\{R\subseteq Q\mid R\cap I\neq\emptyset\}$ is the set of accepting states of the power\=/set automaton. Consider the functional $\fun{\MM}{\DD}{\bR}$ defined as 
\begin{equation}
\MM(\alpha)=\sum_{R\in\widehat{I}} \alpha(R).
\end{equation}
\label{eq:defMM}
Clearly $\MM$ is monotone (by the definition of ${\preceq}$ and the fact that $\widehat{I}$ is upward\=/closed) and continuous. Thus, the following corollary holds.

\begin{fact}
    \label{ft:sums}
    The sequence $M_d\eqdef \MM(\alpha_d)$ is a~decreasing sequence of numbers with the limit $M_\infty\eqdef\MM(\alpha_\infty)$.
\end{fact}

Moreover, Proposition~\ref{pro:alpha-vs-types} implies the following property of the values $M_d$.

\begin{remark}
\label{rem:M-d-and-lang}
For each $d\in\bbN$ we have $M_d=\mu^{d}\big(\lang^d(\AA)\big)$.
\end{remark}

\subsection{The measure of the language}

We are now in position to provide a way of computing the measure $\standardMeasureBig{\lang(\AA)}$ of the considered language $\lang(\AA)$.

\begin{proposition}
    \label{pro:measure-is-fixpoint}
    Using the value $M_\infty$ as defined in Fact~\ref{ft:sums}, we have
    $\standardMeasureBig{\lang(\AA)}=M_\infty$.
\end{proposition}

\begin{proof}
    Consider a~depth $d\in\bbN$ and let
    \[L_d\eqdef\{t\in\treesInfinite{\Gamma}\mid t^{\leq d}\in\lang^d(\AA)\}.\]
    
    Observe that by Lemma~\ref{lem:safety-representation} we know that $\lang(\AA)=\bigcap_{d\in\bbN} L_d$. Moreover, from the definition of $\lang^d(\AA)$ we know that for $d\leq d'$ we have $L_d\supseteq L_{d'}$. Therefore, $\standardMeasureBig{\lang(\AA)}=\lim_{d\to\infty} \standardMeasureBig{L_d}$. However, by the definition of the measure on $\treesInfinite{\Gamma}$, see~\eqref{eq:partial-measure} we know that $\standardMeasureBig{L_d}=\mu^{d}\big(\lang^d(\AA)\big)$ where the latter probability is taken in $\trees{\Gamma}^d$. Now, by Remark~\ref{rem:M-d-and-lang} we know that $\mu^d\big(\lang^d(\AA)\big)=M_d$. Therefore, the thesis holds.
\end{proof}

We are now in place to prove the main theorem of this section.

\begin{proof}[Proof of Theorem~\ref{thm:openSetsUpperBound}]
    By Proposition~\ref{pro:measure-is-fixpoint} we know that $\standardMeasureBig{\lang(\AA)}$ equals $\sum_{R\in\widehat{I}} \alpha_\infty(R)$, where~$\alpha_\infty$ is the greatest fixpoint of the operator $\fun{\FF}{\DD}{\DD}$ (see Lemma~\ref{lem:greatest}). Since the operator $\FF$ is given by a~vector of polynomials in $\bR^{\powerset(Q)}$, one can express in first\=/order logic over $\langle \bR, {+}, {\cdot}\rangle$ that $\alpha_\infty$ is the greatest fixpoint of $\FF$. For the sake of completeness (and complexity analysis) we will provide a~more precise construction here.
    
    Fix a~non\=/deterministic safety automaton $\AA=\langle \Gamma, Q, \delta, I\rangle$. Let $R_1,\ldots,R_N$ be an~enumeration of all the subsets of $Q$. First, consider the formulae describing: the elements of $\DD$ (represented as sequences of reals), upward\=/closed subsets of $\powerset(Q)$ (represented as sequences of zeros and ones), and the order~${\preceq}$:
    \begin{align*}
    \varphi_{\DD}\big(x_1,\ldots,x_N\big) &\eqdef\ \sum_{i=1}^{N} x_i{=}1 \wedge \bigwedge_{i=1,\ldots,N} 0{\leq} x_i{\leq} 1,\\
    \varphi_{\powerset(Q)}\big(\iota_1,\ldots,\iota_N\big) &\eqdef \bigwedge_{i=1,\ldots,N} \iota_i^2 = \iota_i\wedge\bigwedge_{R_i\subseteq R_j} \iota_i\leq \iota_j,\\
    \varphi_{\preceq}\big(\vec{x}, \vec{y}\big) &\eqdef \varphi_{\DD}\big(\vec{x}\big)\wedge \varphi_{\DD}\big(\vec{y}\big) \wedge \forall \vec{\iota}.\ \varphi_{\powerset(Q)}\big(\vec{\iota}\big) \Rightarrow \sum_{i=1}^{N} x_i\iota_i \leq \sum_{i=1}^{N} y_i\iota_i.
    \end{align*}
    As observed above, the operator $\FF$ is just a~vector of polynomials, which means that there exists a~formula $\varphi_{\FF}\big(\vec{x},\vec{f}\big)$ that holds for a~given tuples, if $\vec{x},\vec{f}\in\DD$ and $\FF(\vec{x})=\vec{f}$. Using these, one can write the following formula expressing the measure of the set $\lang(\AA)$:
    \begin{align*}
    \varphi_{\standardMeasureS(\lang(\AA))}(m)\eqdef \exists \vec{x}.\ &\varphi_{\FF}\big(\vec{x},\vec{x}\big)\wedge\\
    &  \left(\forall \vec{y}.\ \varphi_{\FF}\big(\vec{y},\vec{y}\big)\Rightarrow \varphi_{\preceq}\big(\vec{y},\vec{x}\big)\right)\wedge\\
    & \sum_{R_i\in\widehat{I}} x_i = m.
    \end{align*}
    
    Notice that the above formula is of size polynomial in $N$, i.e.~exponential in the number of states and transitions of $\AA$. The following claim follows now directly from Lemma~\ref{lem:greatest}.
    \begin{claim}
    The measure $\standardMeasureBig{\lang(\AA)}$ is the unique real number $m$ satisfying the formula $\varphi_{\standardMeasureS}(m)$.
    \end{claim}
            
    Therefore, by Tarski's Quantifier Elimination~\cite{tarskiDecision1951,collins_algebraic_decomposition}, we know that there exists a~semialgebraic set~\cite[Chapter~2]{bochnak_real_algebraic} containing a~single 
    real number~$\standardMeasureBig{\lang(\AA)}$. This set can be computed in three\=/fold exponential time and can be used as a~representation of $\standardMeasureBig{\lang(\AA)}$.
    
    The complexity comes from results of~\cite{benOrFOonRealsComplexity}: the theory of~real closed fields can be decided in deterministic exponential space. Therefore, the decision problems (i.e.~positivity) about $m=\standardMeasureS(\lang(\AA))$ can be solved in two\=/fold exponential space. 
\end{proof}

As the following remark implies, the use of algebraic numbers in the above procedure is unavoidable.

\begin{remark}
\label{rem:closed-algebraic}
There exists a~regular language recognisable by a~non\=/deterministic safety automaton such that the uniform measure of the language is irrational.
\end{remark}

\begin{proof}
It is enough to notice that the language $L$ from Proposition~\ref{prop:algebraicValueFO} is topologically closed and therefore recognisable by a~non\=/deterministic safety automaton.
\end{proof}

\subsection{Regular languages of finite trees}
\label{sec:finite-trees}

In this section we discuss how the above procedure for closed regular languages can be adjusted to the case of regular languages of finite trees. The main issue in that situation is the fact that for a~fixed alphabet $\Gamma$, the set of all finite trees $\treesFinite{\Gamma}$ is countably infinite. Therefore, there is no ``uniform'' measure over $\treesFinite{\Gamma}$. One of the possible solutions is~to~consider \emph{discounted} measures, i.e.~measures where the structure of~the tree is randomly generated top to bottom, and where at each stage of the process there is some \emph{extinction probability}, i.e.~the probability that we reach \emph{an~end of the structure}.
Such systems, called \emph{branching processes}, have been extensively studied, see e.g.~\cite{branchingProcessesHarris}.
For the study of~their reachability and extinction properties see e.g.~\cite{recursiveStochasticGamesEtessami} and its references.

In our setting an extinction of a~branching process implies generation of a~finite tree. Thus, branching processes that are extinct with probability $1$ generate, in the natural way, a~measure on the set of~finite trees.

For the sake of this section we will formalise that process as function from infinite trees into finite ones.
Let $\Gamma=\{a_1,\ldots,a_n\}$ be an~alphabet of $n$ letters and let $\Gamma'=\Gamma\cup\{\flat_1,\ldots,\flat_n\}$, where $\flat_i$ are distinct symbols not belonging to $\Gamma$. A~tree $t\in\treesInfinite{\Gamma'}$ is called \emph{bounded} if on each branch there is an~occurrence of a~symbol $\flat_i$ for some $i\in\{1,\ldots,n\}$. Notice that analogously to the language $L_{a2}$ from Item~3 of Example~\ref{ex:someMeasuresOfSets}, the measure of bounded trees is $1$. Now, if a~tree $t\in\treesInfinite{\Gamma'}$ is bounded, let $f(t)\in\treesFinite{\Gamma}$ be obtained in the following way:
\begin{align*}
\dom{f(t)}&\eqdef\{u\in\dom{t}\mid \forall v\ancestor u.\ t(v)\in\Gamma\},\\
f(t)(u)&\eqdef \begin{cases}
t(u) & \text{if $t(u)\in \Gamma$,}\\
a_i & \text{if $t(u)=\flat_i$.}
\end{cases}
\end{align*}
The measure on $\treesFinite{\Gamma}$ given by $f{\circ}\standardMeasureS$ will be denoted $\mu^{<\omega}$, i.e.~$\mu^{<\omega}(S)\eqdef \standardMeasureBig{f^{-1}(S)}$ --- notice that being bounded is a~regular property and therefore the set $f^{-1}(S)$ is $\standardMeasureS$\=/measurable.

\begin{thm}
\label{thm:finite-trees}
Given a~regular language of finite trees $L\subseteq\treesFinite{\Gamma}$, one can effectively compute the value $\mu^{<\omega}(L)$ in two\=/fold exponential space, and the value is algebraic.
\end{thm}

\begin{proof}
Consider the following set of full trees:
\[L'\eqdef \big\{t\in\treesInfinite{\Gamma'}\mid \text{if $t$ is bounded then $f(t)\in L$}\big\}.\]
It is easy to observe that $L'$ is a~closed regular set of infinite trees. Moreover, given a~non\=/deterministic tree automaton $\AA$ for $L$ one can construct a~non\=/deterministic safety automaton $\AA'$ for $L'$, and the size of $\AA'$ is polynomial in the size of $\AA$. Since the set of bounded trees has measure $1$, the $\standardMeasureS$\=/measure of $L'$ equals the $\mu^{<\omega}$\=/measure of $L$. Therefore, the result follows from Theorem~\ref{thm:openSetsUpperBound}.
\end{proof}


\section{Complexity of computing the measure}
\label{sec:complexity}

In the previous sections, we have described the algorithms for computing the uniform measure of some classes of~simple sets of infinite trees. 
The upper bounds that follow from the devised algorithms are expressed in Theorems~\ref{thm:FOcomputable}, \ref{thm:CQsAreRational}, and~\ref{thm:openSetsUpperBound}.

\comment{
\begin{thm}[Upper bounds]
    There exists an~algorithm that computes the uniform measure of~a~regular set of~trees $\lang(\mathcal{X})$ 
    \begin{itemize}\itemsep=0pt
        \item in three\=/fold exponential space, if $\mathcal{X}$ is a~first\=/order formula over the signature $\{\rootP, s_{\dL}, s_{\dR}, s\} \cup \Gamma$ ;
        \item in exponential space, if $\mathcal{X}$ is a~Boolean combination of~conjunctive queries;
        \item in two\=/fold exponential space, if $\mathcal{X}$ is a~non\=/deterministic safety automaton.
    \end{itemize} 
\end{thm}

\begin{proof}
The first item is~Theorem~\ref{thm:FOcomputable}; the second Theorem~\ref{thm:CQsAreRational}, and the last one
Theorem~\ref{thm:openSetsUpperBound}.
\end{proof}
}

In the rest of~this section, we~derive lower bounds for the problem of positive measure.
In the first two cases, i.e.~for FO and BCCQ we will reduce directly from the problems of acceptance of Turing machines, while the lower bound in the case of~non\=/deterministic safety automata will be~more direct.

\paragraph*{Turing machines}
A~Turing machine $M$ is a~tuple $\langle \Gamma, Q, \delta, q_\init, q_{f} \rangle$, where $\Gamma$ is an~alphabet, $Q$ is a~finite set of states, $q_\init \in Q$ is an~initial state, $q_{f} \in Q$ is a~final state, and $\delta \subseteq Q \times \Gamma \times Q \times \Gamma \times \{-1,0,1\}$ is the transition relation.
A~run is a sequence of configurations; a~configuration is a~sequence of~length~$2^n$ of cells; and each cell contains either:~$\blank$ denoting an~empty cell,
$a\in \Gamma$ denoting a~non\=/empty cell without the head, or a~pair $\langle q,a\rangle \in Q \times (\Gamma \cup \{\blank\})$
denoting a~cell with the head of the machine over it.
A~run is~\emph{proper} if
\begin{itemize}
  \item every configuration is \emph{proper}, i.e.~there is exactly one head, and after empty cells there are empty cells;
  \item the first configuration is \emph{the initial configuration}, i.e.~the first letter of the configuration is $\langle q_\init, \blank \rangle$;
  \item and \emph{transitions between subsequent configurations are correct}, i.e.~labels of cells without the head do not change in transition and the lablel of the cell with the head of the machine and the new position of the head change accordingly to the transition relation.
\end{itemize}
A~run is accepting if~it is proper and contains a~configuration with the final state (a~\emph{final configuration}).

\subsection{Complexity for FO without descendant}

We begin by~giving a~lower bound in the case of~first\=/order formulae without descendant.

\begin{lemma}
\label{lemma:hardnessFO}
    The positive $\standardMeasure{\text{FO}}$ problem is \expspace\=/hard.
\end{lemma}

To prove this lemma we~reduce the problem whether a~given Turing machine accepts the empty input in exponential space.

\begin{proof}
The following problem is $\expspace$\=/complete. Given a~non\=/deterministic Turing machine $M$ and a~number $n$ in unary, decide whether $M$ accepts the empty input using at~most $2^n$ memory cells.

For a~given Turing machine $M$ and number $n$ we will describe how to construct a~first\=/order formula~$\phi$, such that $\standardMeasure{\lang(\phi)} > 0$
if and only if $M$ accepts the empty input using~at most $2^n$ memory cells.
The idea is that the trees that satisfy the formula~$\phi$ encode all accepting runs of $M$ of the desired length. The size of the formula $\phi$ will be polynomial in the value of $n$ and the size of $M$.
Before we proceed, recall that if~$M$ accepts the empty input using no~more than $2^n$ memory cells, then there exists an~accepting run of~$M$
that uses no~more than~$2^{2^n}$ steps.

A~run of~$M$ will be~encoded as a~set of trees $\setBall{t}$, where $t$ is a~full tree of~height $d= 2^n + n$, such that
every node up~to~the depth $2^{n-1}$ is labelled with a~fresh letter $r \not \in \Gamma$, every node 
at depth $2^n$ is labelled with a~fresh letter $c$, every node between depths $2^{n} + 1$ and $2^{n} + n - 1$, inclusive, is labelled with a~fresh letter $c'$, and every node at~depth $2^{n} + n$ is labelled with a~memory cell content. For an~illustration of such a~tree $t$ see Figure~\ref{fig:FO-lower-bound}.

\begin{figure}
\centering
\newcommand{\spod}{-6.5}
\begin{tikzpicture}[scale=0.62]

	\newcommand{\smalltree}[1]{
		\draw (#1)-- +(-1.9,-1.9)
		(#1) -- +(1,-1.9);
		
	}

    \node (root) at (0,0) {$r$};
    \node (lend) at (-8,-8) {};
    \node (rend) at (8,-8) {};
    
    \node (x1) at (-1,-2) {$r$};
    \node (x2) at (-.5,-3) {$r$};
    \node (x3) at (2.5, -3) {$r$};
    \node (xn) at (-3, -3) {$r$};
    \node (v1) at (-2.5, -4) {$r$};
    \node (v2) at (1.5, -4) {$r$};
    \node (v3) at (2.5, -4) {$r$};
    \node (vn) at (-4,-4) {$r$};
    \node (t1) at (-2,-5) {$c$}; 
    \node (t2) at (1, -5) {$c$};
    \node (t3) at (4.5, -5) {$c$};
    \node (tn) at (-5, -5) {$c$};
	
	\draw (root) -- (xn) -- (vn) -- (tn) --  (lend);
	\draw (root) -- (rend);
	\draw[dashed, thick]
		(root) -- (x1)
		(x1) -- (x2)
		(root) -- (x3)
		(x3) -- (v3)
		(x3) -- +(1,-1)
		(v1) -- (t1)
		(v2) -- (t2)
		(v3) -- (t3);
	\draw[solid, thick]
		(x2) -- (v1)
		(x2) -- (v2);

	\smalltree{t1}
	\smalltree{t2}
	\smalltree{t3}
	\smalltree{tn}

	\draw[dotted] (-10,0) -- (root) -- (10,0);
	\draw[dotted] (-10,-5) -- (tn) -- (t1) -- (t2) -- (t3) -- (10,-5);
  	\draw[dotted] (-10,\spod) -- (10,\spod);
	\draw[<->] (-10,0) -- (-10, -5) node[midway,left] {$2^n$}; 
   	\draw[dotted] (-10,\spod) -- (10,\spod);
  	\draw[<->] (-10,-5) -- (-10, \spod) node[midway,left] {$n$}; 
\end{tikzpicture}
\caption{Encoding using first-order formulae. The bold solid lines denote some of the child relations, the dashed lines denote the induced ancestor relations.
The $c$\=/labelled nodes are roots of encodings of configurations; the $r$\=/labelled nodes span the sequence of configurations. The ``diamond'' between second and third drawn configuration assures that they are consecutive.}
\label{fig:FO-lower-bound}
\end{figure}
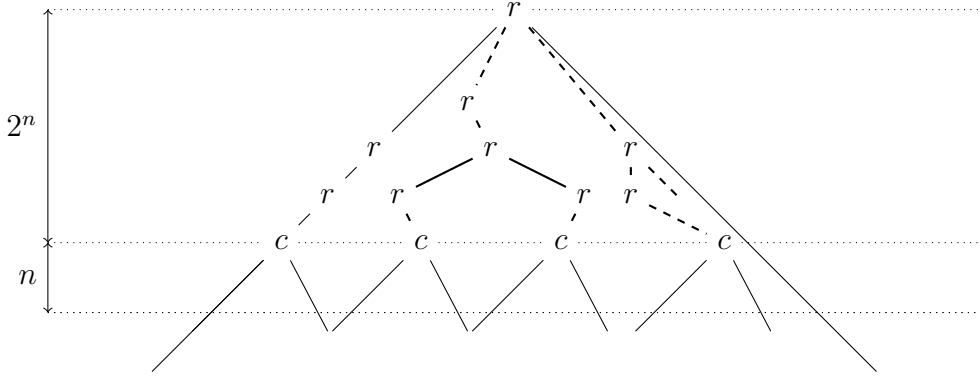


Before we describe the encoding, let us recall a~standard construction in which a~polynomial formula can select two nodes that are in~a~distance that is~exponential in~the size of~the formula.%
\begin{align}
    \varphi_{s}^{0}(x,z) &\eqdef x\ s\ z,\nonumber\\
    \varphi_{s}^{i+1}(x,z) &\eqdef \exists y.\ \forall x',y'.\ (x'{=}x \land y'{=}y){\lor}(x'{=}y\land y'{=}z)\Rightarrow \varphi_{s}^{i}(x',y')\label{eq:fo-exponential-distance}
\end{align}

Intuitively, the formula $\varphi^{n}_{s}(x,z)$ states that $x$ is an~ancestor of~$z$ such that $\td(x,z) = 2^n$.
Let $\varphi^{\td=2^n}_{s}$ denote that formula.
In a similar fashion we can define a~formula $\varphi^{\td<n}_{s}(x,z)$ stating that $x$ is an~ancestor of~$z$ and $\td(x,z) < 2^n$:
we simply replace the inductive construction with%
\begin{align}
    \psi_{s}^{0}(x,z) &\eqdef x\ s\ z \lor x = y,\nonumber\\
    \psi_{s}^{i+1}(x,z) &\eqdef \exists y.\ \forall x',y'.\ (x'{=}x \land y'{=}y){\lor}(x'{=}y\land y'{=}z)\Rightarrow \psi_{s}^{i}(x',y')\label{eq:fo-subexponential-distance}
\end{align}
and write $\varphi^{\td<2^n}_{s}(x,z) \eqdef \exists y.\ \psi^{n}_{s}(y,z) \land y s x \land \lnot z s x$.
In the same manner we define $\varphi^{\td=2^n}_{s_{\dL}}$, $\varphi^{\td<2^n}_{s_{\dL}}$ and $\varphi^{\td=2^n}_{s_{\dR}}$, $\varphi^{\td<2^n}_{s_{\dR}}$.

Now we can return to the encoding.
The sub\=/trees of $t$ in nodes labelled $c$ denote configurations. Since every such sub\=/tree is~of~height $n$, every encoding 
has exactly $2^n$ leafs encoding memory cells. Moreover, in a~tree $t$ there are exactly $2^{2^n}$ sub\=/trees encoding configurations.

The only non\=/trivial part of the first\=/order formula $\varphi_{frame}$ defining the above set of trees is the one talking about the nodes at depth~$2^n$. This is~done using $\varphi^{\td=2^n}_{s}(x,y)$ that expresses the fact that $x\ancestor y$ and $\td(x,y)=2^n$.
Similarly, using $\varphi^{\td<2^n}_{s}$ we~can check the labels of the intermediate nodes and thus verify that a~given tree has the shape as~described above.

Now, all we need is to describe how to write a~formula encoding an~accepting run.
We start with a~formula $\varphi_{proper}$ stating that every configuration is proper, i.e.~there is exactly one head, 
and after empty cells there are empty cells.
Then a~formula $\varphi_{first}$ defining initial configuration: it~simply states that 
in the left\=/most configuration (defined using $\varphi^{\td=2^n}_{s_{\dL}}$) head is over the first memory cell and that the first memory cell is empty.
Next is $\varphi_{accepting}$ defining an~accepting configuration: it simply states that
there is node labelled by a~letter from the set $\{q_{f}\} \times (\Gamma \cup \{\blank\})$.

Now we define a~formula $\varphi_{transition}$ stating that transitions between subsequent configurations are correct.
To do that, we define $\varphi_\textit{next\_conf}(x,y)$, stating that $x$ and $y$ are $c$\=/labelled roots
of two consecutive configurations, and $\varphi_{same\_cell}(x,y)$ stating that nodes $x$ and $y$ encode the same cell in the two 
consecutive configurations. More formally, the formula $\varphi_{same\_cell}(x,y)$ states that there are two sequences of nodes $x_0, \dots, x_n$ and $y_0,\dots, y_n$ such that $y_n = y, x_n = x$, $\varphi_\textit{next\_conf}(x_0,y_0)$, and for $i>0$ $x_i$ is a~left child if and only if $y_i$
is a~left child. Given the formula $\varphi_{same\_cell}(x,y)$, the formula $\varphi_{transition}$ simply states that the labels of $x$ and $y$ are consistent with the transition relation.

What remains is to define the formula $\varphi_\textit{next\_conf}(x,y)$. It simply states that $x$ and $y$ are labelled $c$ and that the path from $x$ to $y$ ``forms a~diamond'', i.e.~that there are three nodes $x_a$, $y_a$, and $z$ such that $x_a$ is an~ancestor of $x$ reachable by $s_{\dR}$ only, $y_a$ is an~ancestor of $y$ reachable by $s_{\dL}$ only, and $x_a$, $y_a$ are respectively the left and the right child of~$z$. To avoid speaking about the ancestors explicitly, we use the formulae $\varphi^{\td<2^n}_{s_{\dL}}$ and $\varphi^{\td<2^n}_{s_{\dR}}$.
 
The resulting formula is the conjunction of the formulae $\varphi_{frame}$, $\varphi_{proper}$, $\varphi_{first}$, $\varphi_{accepting}$, $\varphi_{transition}$.
\end{proof}

\subsection{Complexity for conjunctive queries}

In the case of a~conjunctive queries, we~obtain exact bounds.
More precisely, we observe that deciding whether the measure of~a~language defined by~a~single conjunctive query has a~positive measure is~\np\=/complete.
\begin{proposition}
    The positive $\standardMeasure{\text{CQ}}$ problem is \np\=/complete.
\end{proposition}

\begin{proof}
    Let $q$ be a conjunctive query. Then, either $q$ is not satisfiable and $\standardMeasure{q} = 0$, or
    $q$ is satisfiable and has a~positive measure.
    That is, $\standardMeasure{q} > 0$ if~and only if $q$ is satisfiable.
    Deciding whether a conjunctive query is satisfiable is $\np$\=/complete, cf. e.g.~\cite{bjorklundCQscontainment}.
\end{proof}

If~we~allow Boolean combinations of~conjunctive queries the complexity increases exponentially.

\begin{lemma}
	\label{lemma:hardnessPositiveBCCQ}
    The positive $\standardMeasure{\text{BCCQ}}$ problem is \nexp\=/hard.
\end{lemma}

The proof is a~very simple adaptation of~the reduction in~\cite{murlakAcyclic}, which shows that the satisfiability of~Boolean combinations of~conjunctive queries with respect to~recursion\=/free \emph{DTD}s is \nexp\=/hard.
The two differences are the lack of \emph{DTD}, which is~replaced by a~Boolean combination of~patterns, and the forced use of the root patterns.

\begin{proof}
To show the lower bound we reduce the following $\nexp$\=/complete problem: given a non\=/deterministic Turing machine $M$ and a~number $n$ in unary, decide whether $M$ accepts the empty input in~at most $2^n$ steps.

More precisely, for a~Turing machine $M$ and number $n$ we will describe how to construct a~Boolean combination of rooted, firm conjunctive queries~$\phi$, such that $\standardMeasure{\lang(\phi)} > 0$
if and only if $M$ accepts the empty input in~at most $2^n$ steps.
The idea is that the trees that satisfy the formula~$\phi$ encode all accepting runs of $M$ of the desired length.

The construction will be similar to the one in the used in the proof of~Lemma~\ref{lemma:hardnessFO} with appropriate
changes to the formulae $\varphi_{frame}$, $\varphi_{proper}$, $\varphi_{first}$, $\varphi_{accepting}$, $\varphi_{transition}$.

Since a~run should accept in $2^n$ steps, the machine will not use more than $2^n$ memory cells, thus every configuration under consideration has length at~most $2^n$.

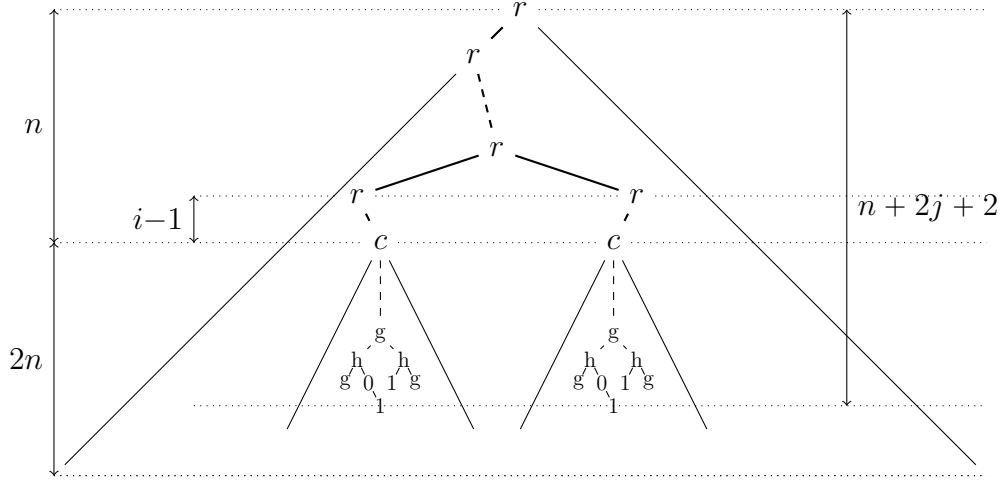
\begin{figure}

\centering
\newcommand{\spod}{-10}
\begin{tikzpicture}[scale=0.62]

\tikzstyle{innerNod} = [scale=0.7]

\newcommand{\smalltree}[1]{
    \draw 
     (#1)-- +(-2,-4)
     (#1) -- +(2,-4);
     
     \node[innerNod] (c) at ($(#1) + (0,-2)$) {g};
     \node[innerNod] (cL) at ($(c) + (-.5,-.5)$) {h};
     \node[innerNod] (cR) at ($(c) + (.5,-.5)$) {h};     
     
     \node[innerNod] (cLL) at ($(cL) + (-.25,-.5)$) {g};
     \node[innerNod] (cLR) at ($(cL) + (.25,-.5)$) {0};     
     \node[innerNod] (cLRR) at ($(cLR) + (.25,-.5)$) {1};    
     
     \node[innerNod] (cRL) at ($(cR) + (-.25,-.5)$) {1};
     \node[innerNod] (cRR) at ($(cR) + (.25,-.5)$) {g};           
     
     \draw
     (c) -- (cL)
     (c) -- (cR)
     (cL) -- (cLL)
     (cL) -- (cLR)
     (cLR) -- (cLRR)
     (cR) -- (cRL)
     (cR) -- (cRR);
     
     \draw[dashed] (#1) -- (c);         
}

\node (root) at (0,0) {$r$};
\node (lend) at (-10,-10) {};
\node (rend) at (10,-10) {};

\node (x1) at (-1,-1) {$r$};
\node (x2) at (-.5,-3) {$r$};
\node (v1) at (-3.5, -4) {$r$};
\node (v2) at (2.5, -4) {$r$};
\node (t1) at (-3,-5) {$c$}; 
\node (t2) at (2, -5) {$c$};

\draw (root) -- (x1) -- (lend);
\draw (root) -- (rend);
\draw[dashed, thick]
(x1) -- (x2)
(v1) -- (t1)
(v2) -- (t2);
\draw[solid, thick]
(root) -- (x1)
(x2) -- (v1)
(x2) -- (v2);

\smalltree{t1}
\smalltree{t2}

\draw[dotted] (-10,0) -- (root) -- (10,0);
\draw[dotted] (-10,-5) -- (t1) -- (t2) -- (10,-5);
\draw[dotted] (-10,\spod) -- (10,\spod);
\draw[<->] (-10,0) -- (-10, -5) node[midway,left] {$n$}; 
\draw[dotted] (-10,\spod) -- (10,\spod);
\draw[<->] (-10,-5) -- (-10, \spod) node[midway,left] {$2n$}; 

\draw[dotted] (-7,-4) -- (v1) -- (v2) -- (10,-4);
\draw[<->] (-7,-4) -- (-7, -5) node[midway,left] {$i{-}1$}; 

\draw[dotted] (-7,-8.5) -- (-3, -8.5) -- (2, -8.5) -- (10,-8.5);
\draw[<->] (7,0) -- (7, -8.5) node[midway,right] {$n + 2j + 2 $}; 
\end{tikzpicture}
\caption{The structure of the encoding using conjunctive queries. The figure depicts two consecutive configurations, with an~ancestor $i$ levels above. The formula $\varphi_{\text{same\_turn}_j}(x,y)$ chooses either both left $h$\=/labelled nodes or both right $h$\=/labelled nodes.}

\label{fig:CQ-lower-bound}
\end{figure}

A~run will be encoded as follows, see Figure~\ref{fig:CQ-lower-bound}. Let $t'$ be a~full tree of~height~$d+1$.
The root of $t'$ is labelled $r\notin\Gamma$ and the tree $t'.\dR$ has every node labelled $w\notin\Gamma$.
Now we describe the tree $t = t'.\dL$. The root of $t$ is labelled $r$. And every node of $t$ at depth $l$ such that $1 \leq l \leq n-1$ is labelled $r$.
Every node of $t$ at depth $n$ is labelled $c\notin\Gamma$. The trees rooted in the $c$-labelled nodes encode configurations.
Children of the $c$-labelled nodes are labelled with $g$. Every $g$-labelled node $u$ at depth at most $d-3$ has $h$-labelled children ($h\notin\Gamma$).
The nodes $u\dL\dL$ and $u\dR\dR$ are labelled $g$; the node $u\dL\dR$ is labelled $0$; the node $u\dL\dR\dR$ is labelled~$1$; 
and the node $u\dR\dL$ is labelled $1$.
Nodes $u$ at depth $d-2$ labelled $g$ have similar setup with the difference that the nodes $u\dL\dL$ and $u\dR\dR$ are labelled 
by a letter from the alphabet $\Gamma^{t}$: they encode the memory cells. The unspecified nodes are labelled with $w$.

Notice that for $d=3n$, we have $2^n$ configurations, nodes labelled $c$, with $2^n$ memory cells each, and that
all those requirements can be easily forced by a~polynomial in size Boolean combination of patterns.
Indeed, the above frame can be expressed by a Boolean combination $\varphi_{frame}$, in which the used patterns
forbid ill-labelled children and enforce good labelling of the roots.
Thus, all that is left to define is~a~Boolean combination of patterns that will enforce that a~tree as above 
encodes an accepting run. We~will do this by stating that the configurations are proper ($\varphi_{\textit{proper}}$); that there is a configuration with an accepting state $\varphi_{\textit{accepting}}$; that
the first configuration is the initial configuration ($\varphi_{\textit{first}}$); and that each transition between two consecutive configurations 
is consistent with the transition relation $\delta$ ($\varphi_{\textit{transition}}$). The first two requirements are easily expressible, it is also easy to write a~formula saying that every configuration is proper, i.e.~that there is only one head and after blanks there are blanks.

Now, we describe how to express that the transitions are proper.
We start with family of formulae $\varphi_{\text{next\_conf}_{i}}(x,y)$, $i\in\{1,2,\dots, n\}$, where $\varphi_{\text{next\_conf}_{i}}(x,y)$ holds for two nodes
that are roots of two consecutive configurations with a common ancestor $i$ levels above.
The formula $\varphi_{{\text{next\_conf}_i}(x,y)}$ says that $x$ and $y$ are $c$-labelled and 
there are two sequences $x_1, \dots, x_{i-1}, x$ and $y_1, \dots, y_{i-1}, y$, and a node $z$ such that
$x_1$ is the left child of $z$, $y_1$ is the right child of $z$, $x_1, \dots, x_{i-1}, x$ are right-child connected 
$y_1, \dots, y_{i-1}, y$ are left-child connected, they ``form a diamond'' as on~Figure~\ref{fig:CQ-lower-bound}.
In a similar way, we can define $\varphi_{\text{next\_cell}_i(x,y)}$ choosing nodes that are two consecutive memory cells in the same configuration.

Now, we will describe how to construct a family of formulae $\varphi_{\text{same\_cell}_i}(x,y)$, $j\in\{1,2,\dots, n\}$, that chooses the same memory cell in two consecutive $\varphi_{\text{next\_conf}_i}$ configurations.
We start by observing that if $x_c$ and $y_c$ are the $c$-labelled nodes of the appropriate configurations then $x$ and $y$ are the same memory cell if at each level of the path from $x_c$ to $x$ we choose the same turns as on the appropriate levels of the path from $y_c$ to $y$.
In other words, if $x_c, x_1, \dots, x$ is child connected path from $x_c$ to $x$ and $y_c, x_y, \dots, y$ is child connected path from $y_c$ to $y$, then $x_i$ is the left child of its parent if and only if $y_i$ is the left child of its parent. This can be achieved
by a family of formulae $\varphi_{\text{same\_turn}_j}(x,y)$.

The formula $\varphi_{\text{same\_turn}_j}(x,y)$ states that $x$ and $y$ are at depth $2j+n$, are $h$-labelled, and 
are both left or both right children. We achieve the last one by~stating that there are two sequences of nodes $x_1, \dots x_{n+2j+1}$ and $y_1, \dots y_{n+2j+1}$ such that for $j=1,\dots, n+2j+1$ we have that
$x_j s x_{j+1}$,$y_j s y_{j+1}$, $x$ is an ancestor of $x_{n+2j+1}$, $y$ is an ancestor of $y_{n+2j+1}$,
$x_{1}= y_{1}$ and both $y_{n+2j+1}$, $x_{n+2j+1}$ are labelled with $1$. Note, that $\varphi_{\text{same\_turn}_i}(x,y)$ if and only if $x$ and $y$ are on the same depth and both are left or both are right children.

Finally, we define $\varphi_{\text{same\_cell}_i}(x,y)$ by stating that $x$ and $y$ are at depth $d$ and
the paths from their respective $c$-labelled ancestors take the same turns.

With $\varphi_{\text{same\_cell}_i}(x,y)$ defined, for every transition $\sigma$ that is not in $\delta$
we can write a~pattern $\varphi_{\sigma}$ stating that a~transition between two configurations is incorrect. 
The pattern states that for some $i,j,k$ there are two nodes $x,y$ such that $\varphi_{\text{same\_cell}_i}(x,y)$ and there also are two nodes $y', y''$ such that 
$\varphi_{\text{next\_cell}_j}(y,y')$, $\varphi_{\text{next\_cell}_k}(y',y'')$ and the labels of $x,y,y',y''$ encode $\sigma$.
Hence, $\varphi_{transition}$ is simply the conjunction of the negations of the formulae $\varphi_{\sigma}$.

The final formula is the conjunction of the formulae $\varphi_{frame}$, $\varphi_{proper}$, $\varphi_{first}$, $\varphi_{accepting}$, $\varphi_{transition}$, 
modified so that the used patterns are firm. We do that by quantifying every variable by formula $\varphi_{\text{limited\_depth}_d}(x)$ stating
that there is a~node $y$ at depth $d+1$ such that $x$ is an ancestor of $y$.
\end{proof}

Since the lower bound from Lemma~\ref{lemma:hardnessPositiveBCCQ} matches the upper bound obtained in Proposition~\ref{prop:upperPositiveBCCQ} we infer the following.

\begin{thm}
    \label{thm:completenessBCCQ}
    The positive $\standardMeasure{\text{BCCQ}}$ problem is \nexp\=/complete.
\end{thm}

\paragraph*{Counting complexity}

We will now provide a~short explanation how the above results can be read in terms of counting complexity classes. However, as the main scope of the article is expressed in terms of classical decision procedures, the provided explanation is brief, only indicating how this approach can be~taken.

A~careful inspection of the proof of~hardness, i.e.~the proof of~Lemma~\ref{lemma:hardnessPositiveBCCQ}, reveals that the described reduction in fact reduces the problem of~computing the number of~runs of~a~non\=/deterministic Turing machine running in~exponential time to~the problem of computing the standard measure of~a~language defined by a~Boolean combination of~conjunctive queries. The crucial observation is that under the above defined reduction, every accepting run of~the machine defines an~unique prefix of fixed depth, say~$k$. Moreover, each such run corresponds to exactly the set of trees that conform to that prefix. Since for every run the size of the corresponding prefix is $2^{k}$, there exists a~computable constant~$c$ such that the number of~accepting runs, say~$a$, and the measure, say~$\mu$, are related by the equation $c\cdot a = \mu$. Thus, it shows that the quantitative $\standardMeasure{\text{BCCQ}}$ problem is $\sharp \exptime$\=/hard.

On the other hand, a~careful analysis of the proofs of Theorem~\ref{thm:CQsAreRational} and Corollary~\ref{cor:booleanCombinationOfCQsAreRational} 
allow us~to~claim that computing the measure belongs to~the class $\sharp \exptime$.
Let $k$ be~the size of the input query. A~simple algorithm guesses a~full binary tree of height~$k$ that is~the prefix
witnessing the satisfiability of the root patterns and then verifies it~in deterministic exponential time.
Since the size of every such prefix is exponential in~$k$ and the prefixes define a~partition into sets of equal measure, instead of counting the prefixes, it~is~enough to~count the accepting runs.

Thus, we state the following.

\begin{proposition}
\label{thm:complexityCountingBCCQ}
The problem of computing the measure of $\text{BCCQ}$ (i.e.~Boolean combinations of conjunctive queries) is $\sharp \exptime$\=/complete.
\end{proposition}

Notice that the arguments above use crucially the idea that in the case of~conjunctive queries
the measure is determined by the close neighbourhood of the root.
This allows us~to~simply enumerate all the prefixes of~some fixed depth
and count the positive occurrences.
A~similar observation can be made in the case of~first\=/order formulae not using descendant relation, but
most likely does not extend to~the case of~first\=/order formulae using descendant nor to~the case of~safety automata. Since both classes admit languages of~irrational uniform measure and every prefix of fixed depth is~of~rational measure, simple counting is not enough.

\subsection{Complexity for safety automata}

In this section we provide a~lower bound for the problem of computing the measure of a~language given by a~non\=/deterministic safety automaton. In that case, contrary to the cases of FO and BCCQ, the given model is asymmetric --- the considered model is not closed under complement. Therefore, the problems of positive measure and full measure (i.e.~whether $\standardMeasureBig{L}=1$) are not obviously inter\=/reducible. Therefore, to obtain a~better lower bound we will focus on the later problem.

\begin{proposition}
The problem whether $\standardMeasureBig{\lang(\AA)}=1$ for a~non\=/deterministic safety automaton $\AA$ is \exptime\=/hard.
\end{proposition}

\begin{proof}
The result follows directly from \exptime\=/hardness of the problem of \emph{universality} of non\=/deterministic automata over finite trees (i.e.~whether $\lang(\AA)=\treesFinite{\Gamma}$),
see e.g.~\cite[Theorem~1.7.7]{tata2007}, and the reduction provided in Theorem~\ref{thm:finite-trees}.
\end{proof}

\comment{
That is, the above reduction takes a~non\=/deterministic Turing machine $\MM$ and a~number~$n$ and returns
a Boolean combination $\phi$ of conjunctive queries such that there exists a constant $c$ computable in polynomial time
such that $c^{-1} \cdot \lang(\phi)$  is the number of accepting runs of $\MM$ that use at most $2^n$ memory cells.

In other words, the reduction shows that computing measure of set defined by a~conjunctive query is as hard as counting the number of accepting runs of a non\=/deterministic Turing machine using exponential space.
}


\section{Conclusions and future work}
\label{sec:conclusions}

We have shown that there exists an~algorithm that, given a~first\=/order sentence $\varphi$ over the signature $\Gamma \cup \{ \rootP, s_{\dR}, s_{\dL}, s\}$,
computes $\mu^{\ast}(\lang(\varphi))$ in three\=/fold exponential space.
We also have shown that there exists an~algorithm that, given a~Boolean combination of conjunctive queries $\varphi$ over the signature $\Gamma \cup \{ \rootP, s_{\dR}, s_{\dL}, s, \ancestor \}$, computes the uniform measure $\mu^{\ast}(\lang(\varphi))$ in exponential space. Both these algorithms base on the fact that the measure of the considered language coincides with the measure of an~effectively computable clopen regular set of trees (see Remark~\ref{rem:fo-as-clopen} and the proof of Theorem~\ref{thm:CQsAreRational}). Such languages are always of computable rational measure and thus the results follow. As witnessed by Proposition~\ref{prop:algebraicValueFO} this technique cannot work for the full first\=/order logic, because it can define languages of irrational measures.

Additionally to the above results, we provide an~algorithm for computing the measure of a~language given by a~non\=/deterministic safety automaton. The approach in that case must be different, as witnessed by Remark~\ref{rem:closed-algebraic}. The provided algorithm reduces the problem to a~first\=/order formula over the reals and then invokes the celebrated Tarski's quantifier elimination procedure. This result can be adjusted to the case of regular languages of finite trees, as discussed in Section~\ref{sec:finite-trees}.

We have also studied the computational complexity of the problems of positive measure. The upper and lower bounds match in the case of conjunctive queries (\np\=/complete) and Boolean combinations of conjunctive queries (\nexp\=/complete). The bounds in the case of first\=/order logic without descendant do not match: the problem is \expspace\=/hard while the provided algorithm runs in three\=/fold exponential space. Similarly, the bounds for non\=/deterministic safety automata do not match. We think that establishing the exact bounds of the problems poses~an~interesting direction of~future research.

The substantial gap in the case of first\=/order formulae stems from two reasons. The fist reason is applicable not only in the case of first\=/order formulae and can be~stated as~follows: the provided algorithms solve the computational quantitative problem, i.e.~\emph{compute the measure}, while the provided reductions utilise decision versions of~the qualitative problem, i.e.~\emph{decide whether the measure is~$0$} or~\emph{decide whether the measure is~$1$}. This is~especially evident in the case of Boolean combinations of conjunctive queries, where the problem of deciding positive measure is~$\nexp$\=/complete, computing the measure is~$\sharp \exptime$\=/complete, and we~provide an~algorithm that computes the measure in exponential space. 
To further support this observation, note that in the cases where we~have obtained matching complexity bounds we~have provided specific algorithms to~solve the positive measure problems.

The second reason is more algorithm specific. The proposed algorithm uses the Gaifman locality to compute a~new formula in Gaifman normal form.
This results in~necessarily three\=/fold exponential blow\=/up in~the size of the formula, see~\cite{modelTheoryMakesFormulasLarge}.
On the other hand, the transformation does not alter the bound on~the height of~the counted prefixes, as it remains exponential in~the size of~the original formula.
Notice that our reduction exploits the blow\=/up in the prefix size, but not in~the size of~the new formula.
We conjecture that the actual complexity of the problem should match the lower bound, as we think that it is~possible to circumvent the blow\=/up in the formula
size by performing the translation on\=/line and computing only the necessary parts of the normal form on~demand.

The results of this paper are expressed in terms of trees with binary branching. The results directly extend to trees of other fixed branching, i.e.~ternary trees. The situation is a bit more subtle with $\omega$\=/branching trees: although they can be naturally interpreted into binary trees, the interpretation is not definable in all the considered formalisms (e.g.~in FO with successor). Even worse, it seems that in the case of~unbounded finitely branching trees there is~no~natural definition of an~uniform measure: intuitively, every measure on such trees has to be~effectively discounted, in the sense that having more children has to be less probable than having less. Due to these difficulties, we decided to focus on the binary formulations. If there is a~need of translating them to some other model, one should carefully check whether it is possible to interpret that model within binary branching trees in a~way definable in the considered logic.

Note that the considered measure respects a~form of a~$0{-}1$-law. By Lemma~\ref{lemma:basicMeasuresLemma}, if $t$ is a finite tree then with probability~$1$ 
it appears as a~sub-tree in a~random tree.
It would be interesting to extend the enquiry to measures that do not posses such a property.
Such measures can be expressed, for example, by graphs or by branching boards, cf.~\cite{przybylkoRegularBranchingGamesMFCS}.

Obviously, the most interesting problem is to find an~algorithm that can compute the uniform measure of an~arbitrary regular language of~infinite trees.
While we know that languages with irrational measures exist, we conjecture that for any regular language of trees $L$
the uniform measure $\standardMeasure{L}$ is algebraic.

\bibliographystyle{elsarticle-num} 
\bibliography{przybylko}

\end{document}